\documentclass[pra,aps,superscriptaddress,twocolumn,letter,nopacs,nofootinbib,longbibliography]{revtex4-2}
\usepackage[english]{babel}
\usepackage[utf8]{inputenc}
\usepackage[T1]{fontenc}

\usepackage{amsthm}
\usepackage{amsmath,amssymb}
\theoremstyle{plain} 
\usepackage{graphicx}
\usepackage[caption=false]{subfig} 
\usepackage{float} 
\usepackage[colorlinks=true, allcolors=blue]{hyperref}
\usepackage{braket}
\usepackage{float}
\usepackage{bbm,amsfonts}
\usepackage[shortlabels]{enumitem}
\usepackage{hyperref}
\usepackage{xcolor}
\usepackage{mathtools}
\usepackage[title]{appendix}
\usepackage{dsfont} 
\usepackage{tikz}
\usepackage{rotating} 
\usepackage{mathtools} 

\hypersetup{
	colorlinks=true,  
	linkcolor=blue,   
	citecolor=blue,   
	urlcolor=blue     
}

\newtheorem{theorem}{Theorem}

\newtheorem*{example*}{Example}
\newtheorem{lemma}[theorem]{Lemma}

\newtheorem*{remark*}{Remark}
\newtheorem{corollary}[theorem]{Corollary}

\newtheorem{proposition}[theorem]{Proposition}

\newtheorem{manualtheoreminner}{Theorem}
\newenvironment{manualtheorem}[1]{%
  \renewcommand\themanualtheoreminner{#1}%
  \manualtheoreminner
}{\endmanualtheoreminner}
\newtheorem{manuallemmainner}{Lemma}
\newenvironment{manuallemma}[1]{%
  \renewcommand\themanuallemmainner{#1}%
  \manuallemmainner
}{\endmanuallemmainner}

\begin{document}

\title{Operation fidelity explored by numerical range of Kraus operators} 
\author{Igor Che{\l}stowski}
\affiliation{Faculty of Physics, University of Warsaw, ul. Pasteura 5, 02-093 Warszawa, Poland}
\author{Grzegorz Rajchel-Mieldzio\'c}
\affiliation{ICFO-Institut de Ciencies Fotoniques, The Barcelona Institute of Science and Technology, Av. Carl Friedrich Gauss 3, 08860 Castelldefels (Barcelona), Spain}
\affiliation{Center for Theoretical Physics, Polish Academy of Sciences, Al. Lotnik\'ow 32/46, 02-668 Warszawa, Poland}
\author{Karol {\.Z}yczkowski}
\affiliation{Institute of Theoretical Physics, Jagiellonian University, ul. {\L}ojasiewicza 11, 30–348 Krak\'ow, Poland}
\affiliation{Center for Theoretical Physics, Polish Academy of Sciences, Al. Lotnik\'ow 32/46, 02-668 Warszawa, Poland}

\date{Jul 11, 2023}

\begin{abstract}
Present-day quantum devices require precise implementation of desired quantum channels. 
To characterize the quality of implementation one uses the average operation fidelity $F$, defined as the fidelity between an initial pure state and its image with respect to the analyzed operation, averaged over an ensemble of pure states.
We analyze the statistical properties of the operation fidelity for low-dimensional channels and study its extreme values and probability distribution.
These results are obtained with the help of the joint numerical range of the set of Kraus operators representing a channel.
Analytic expressions for the density $P\left(F\right)$ are derived in some particular cases including unitary and mixed unitary channels as well as quantum maps represented by commuting Kraus operators. 
We propose a scheme for fidelity-based channel estimation which uses the distribution of operation fidelity. 
\end{abstract}

\maketitle
\section{Introduction}
    Advancement of quantum technologies promises development of quantum devices that shall surpass classical counterparts in realizing certain tasks, e.g.\ factorizing numbers, unstructured search, or performing Fourier transform~\cite{Nielsen_2000}.
    However, in order to construct reliable quantum apparatus, it is necessary to overcome noise-inducing state preparation and measurement errors~\cite{Lidar_2013}. 
    
    The full treatment of these errors requires scientists to understand the relationship between ideal quantum channels $\Phi(\rho)$ and their real-world approximations $\tilde{\Phi}(\rho)$~\cite{Junan_2021,Qiu_2023}.
    One way of quantifying this connection for an invertible channel $\Phi$ is to investigate the resulting composed channel $\Theta(\rho) \coloneqq \Phi^{-1} \circ \tilde{\Phi}(\rho)$. 
    If the realization of the initial channel is perfect, then $\Theta$ is an identity channel, $\Theta(\rho) = \rho$. 
    
    Suppose the realization is no longer ideal; then, the difference between the channel $\Theta$ and an identity channel can be assessed by the operation fidelity, i.e.\ the fidelity $F$ between an initial state $\rho$ and the final state $\Theta(\rho)$. 
    For an identity channel, this fidelity is equal to 1, independently of the state $\rho$.
    Notwithstanding, in the general case, the fidelity can be smaller than unity and will be dependent on the state $\rho$. 
    Therefore, to describe the operation fidelity it is necessary to focus on its statistical characterization, such as the mean, the probability density function (PDF), or its lower/upper bounds. 
    Such statistical consideration may prove useful for channel estimation and discrimination~\cite{Pirandola_2019,Krawiec_2020}, as we shall present in this work by providing an explicit setup. 
    
    Fidelity between quantum states was widely investigated as a measure of distinguishability of quantum states~\cite{Jozsa_1994,Bowdrey_2002,Zyczkowski_2005}.
    Study of the operation fidelity, $F=\bra{\psi}\Phi\left(\ket{\psi}\bra{\psi}\right)\ket{\psi}$, started in 2002 with Nielsen's paper showing the relation between the average fidelity $\left\langle F\right\rangle$ of a quantum channel $\Phi$ and its Kraus operators~\cite{Nielsen_2002}. 
	Then, the variance and higher moments of operation fidelity $F$ were studied~\cite{Magesan_2011}.
	However, the resulting formulas for higher moments are not practical due to their complexity.
	The full probability distribution of unitary qubit channels was provided by Pedersen \emph{et al.}~\cite{Pedersen_2008}, as well as a more general case of trace non-preserving operations $\ket{\psi} \mapsto N\ket{\psi}$, given by normal matrices $N$. 
	Furthermore, it was considered under which conditions two different channels can have the same operation fidelity~\cite{Magesan_2011v2}. 
		
    Maximum operation fidelity of generalized Pauli channels has been found by Siudzi{\'n}ska~\cite{Siudzinska_2019}.
    In the case of a general quantum channel, the operational fidelity was studied by Johnston and Kribs using the Choi matrix of the channel~\cite{Johnston_2010,Johnston_2011v2,Johnston_2011}. 
    Their approach proved fruitful in finding approximate expressions for the minimal operation fidelity via eigenvalues of a related matrix and through semidefinite programming.
    Our approach is different as we have found exact formulae for the extremal operation fidelity in specific cases.
    
    Expanding upon the existing literature -- it was also proposed that a genetic algorithm might help to determine minimal operation fidelity~\cite{Gregor_2018}. 
    Optimal strategies to evaluate average operation fidelity were analyzed in~\cite{Reich_2013}.
    Finally, the asymptotic behavior of the operation fidelity as the number of copies of the system increases is shown to be lower-bounded by the $\infty$-norm of the channel~\cite{Ernst_2017}.
    
    The aim of the present investigation is to fill in the missing blocks regarding statistical distribution and bounds on the operation fidelity of quantum gates in special cases.
    Our approach to the problem is characterized by investigating the connection of operation fidelity to joint numerical ranges~\cite{Bonsall_1971}.
    The numerical range of a matrix $A$ is the set of all possible complex values $\bra{v}A\ket{v}$ among normalized vectors $\ket{v}$.
    Likewise, the joint numerical range $W$ of multiple matrices $A_1$,...,$A_n$ is a set of $n$-tuples $(\bra{v}A_1\ket{v},...,\bra{v}A_n\ket{v})$ among all normalized states $\ket{v}$.
    The general rationale of our setup consists in dividing the Kraus operators of a given channel into Hermitian and anti-Hermitian parts.
    To evaluate the distribution $P(F)$ of the operation fidelity of a channel acting on a $d$-dimensional system one needs to perform an averaging over the $2(d-1)$-dimensional set of quantum pure states, which reduces to the Bloch sphere for $d=2$. 
    We demonstrate that it is convenient to replace such an integral by the integration with a suitable measure over the joint numerical range $W$ of the Hermitian and anti-Hermitian parts of the Kraus operators $K_j$ associated with the investigated operation.
		
	The paper is organized as follows. 
	First, in Subsection~\ref{subsec:basic_definitions} we introduce the necessary definitions and prove the connection between the joint numerical ranges and the operation fidelity in Subsection~\ref{subsec:applications_of_numerical_range}.
	Then, Subsection~\ref{subsection:probabilistic_estimation} explains the motivation by introducing a scheme for probabilistic channel estimation.
    In Section~\ref{sec:PDF_for_qubit_qutrit}, we derive simple expressions for the distribution of the operation fidelity for exemplary single qubit channels.
    Later on, we show how the notions of the numerical range and the numerical shadow are useful to obtain analogous distributions for unitary gates acting on a $d=3$ system. 
	Subsequent Section~\ref{sec:mixed_unitary} dwells on the study of mixed unitary qubit channels. 
    In Section~\ref{sec:diagonal_Kraus}, we study channels with diagonal Kraus operators. 
    By using a quadratic form of the operation fidelity we find its global minimum.
    Finally, in Section~\ref{sec:proposed_method}, using the joint numerical shadows of commuting Hermitian operators we find the distribution of the operation fidelity for Schur channels (Subsection~\ref{subsec:Schur_channels}).
    This method allows us to establish an analogous result for an arbitrary unitary channel in Subsection~\ref{subsec:arbitrary_unitary_channel}.
	
\section{Operation fidelity, numerical range and estimation of channels}\label{sec:numerical_range_fidelity}
	In this section, we shall formulate propositions that lay ground for our investigation of the connection between the numerical range and the operation fidelity of quantum channels.
	
    \subsection{Numerical range of an operator}\label{subsec:basic_definitions}
        For~a~linear operator $L$ acting on~Hilbert space $\mathds{C}^d$, where $d\in\mathds{N}$, the numerical range $W\left(L\right)$ is defined as
        \begin{equation}
            W\left(L\right)=\left\{\left\langle\psi\right|L\left|\psi\right\rangle \mathrm{,where}\;\left|\psi\right\rangle\in\mathds{C}^d,\left\langle\psi|\psi\right\rangle=1\right\}.
        \end{equation}
        The numerical range is a~subset of~the~complex plane $\mathds{C}$, which can be treated as $\mathds{R}^2$. 
        Any linear operator $L$ acting on~$\mathds{C}^d$ can be expressed as a~sum of~Hermitian and~anti-Hermitian parts; $L=H+\mathrm{i}A$, where $H$ and~$A$ are Hermitian. 
        This means that $W\left(L\right)$ can be equivalently written as
        \begin{widetext}
        \begin{equation}
            W\left(L\right)=W\left(H,A\right)=\left\{\left(\left\langle\psi\right|H\left|\psi\right\rangle,\left\langle\psi\right|A\left|\psi\right\rangle\right)\in\mathds{R}^2\mathrm{,where}\;\left|\psi\right\rangle\in\mathds{C}^d,\left\langle\psi|\psi\right\rangle=1\right\},
        \end{equation}
        which can be easily generalized to the joint numerical range of a collection $\vec H$ of $n$ Hermitian operators,  $\vec H =
        \{H_1,H_2,\ldots,H_n \}$ in~the~following way:
        \begin{equation}
            W ({\vec  H } )= 
            \left\{\left\langle\psi\right|\vec{H}\left|\psi\right\rangle\in\mathds{R}^n\mathrm{,where}\;\left|\psi\right\rangle\in\mathds{C}^d,\left\langle\psi|\psi\right\rangle=1\right\},
        \end{equation}
        \end{widetext}
        with $\left\langle\psi\right|\vec{H}\left|\psi\right\rangle=\left(\left\langle\psi\right|H_1\left|\psi\right\rangle,\left\langle\psi\right|H_2\left|\psi\right\rangle,\ldots,\left\langle\psi\right|H_n\left|\psi\right\rangle\right)$.
        
        Following Refs.~\cite{Bonsall_1971,Gustafson_1997}, let us recall the definitions of the joint numerical radius $w$ and the Crawford number $c$ of a collection of $n$ Hermitian operators $H_j$ as 
        \begin{equation*}
            \begin{split}
                w(\vec{H})&=\sup\left\{\left|\vec{r}\right|\mathrm{,where}\;\vec{r}\in W(\vec{H})\right\},\\
                c(\vec{H})&=\inf\left\{\left|\vec{r}\right|\mathrm{,where}\;\vec{r}\in W(\vec{H})\right\},
            \end{split}
        \end{equation*}
        where $\left|\cdot\right|$ denotes the Euclidean norm.
    
        The numerical shadow $P_L\left(z\right)$ of~a~linear operator $L$ acting on~Hilbert space $\mathds{C}^d$ is defined as the probability distribution induced on~the~numerical range $W\left(L\right)$ by~the~unique unitarily invariant probability distribution $\mu\left(\psi\right)$ on~the~space of~all~pure states $\Omega_d$~\cite{Gallay_2012,Dunkl_2011}. 
        In~other words, it represents the~likelihood that a~randomly chosen pure state has an expectation value of~$L$ equal to~any~given complex number $z$. It~is given by~the~following general expression,
        \begin{equation}
            P_L\left(z\right)=\int_{\Omega_d}\delta\left(z-\left\langle\psi\right|L\left|\psi\right\rangle\right)\mathrm{d}\mu\left(\psi\right),
        \end{equation}
        where $\delta$ denotes the standard Dirac delta.
        In analogy to the numerical range, one can also introduce the joint numerical shadow of~a~collection of~$n$ Hermitian operators~\cite{Gutkin_2013}:
        \begin{equation}\label{eq:def_joint_numerical_shadow}
            P_{H_1,H_2,\ldots,H_n}\left(\vec{r}\right)=\int_{\Omega_d}\delta\left(\vec{r}-\left\langle\psi\right|\vec{H}\left|\psi\right\rangle\right)\mathrm{d}\mu\left(\psi\right),
        \end{equation}
        where $\vec{r}$ is a~position in~the~$n$-dimensional space in~which $W(\vec{H})$ is embedded.
        
    \subsection{Applications of numerical range for~quantum channel fidelity}\label{subsec:applications_of_numerical_range}
    Before we introduce the applications of the numerical range, let us define an equivalence relation in the set of all channels. 
    We say that two channels belong to the same equivalence class if the distribution of their operation fidelity is the same.
    This definition is non-trivial in the sense that two distinct quantum channels can have the same distribution of operation fidelity (see Appendix~\ref{app:same_fidelity}).

    Let $\Phi$ be a~quantum channel acting on a system of dimension $d$ and defined by a collection of $m$ Kraus operators $K_j$,
        \begin{equation}\label{eq:definition_Kraus}
			\Phi\left(\rho\right)=\sum_{j=1}^m K_j\rho K^\dagger_j,
		\end{equation}
  which satisfy the trace preserving condition: 
  $\sum_j K_j^{\dagger} K_j={\mathbbm I}$.
   The operation fidelity with respect to the initial pure state $\ket{\psi}$ is then defined as
		\begin{equation}\label{eq:fidelity_Kraus}
	        F\left(\Phi,\left|\psi\right\rangle\right)=\left\langle\psi\right|\Phi\left(\left|\psi\right\rangle\left\langle\psi\right|\right)\left|\psi\right\rangle=\sum_{j=1}^m\left|\left\langle\psi\right|K_j\left|\psi\right\rangle\right|^2.
		\end{equation}
		Each Kraus operator can be decomposed into Hermitian and anti-Hermitian parts: $K_j=H_j+\mathrm{i}A_j$, where all $H_j$ and~$A_j$ are Hermitian. Then
		\begin{equation}
	        F\left(\Phi,\left|\psi\right\rangle\right)=\sum_{j=1}^m\left\langle\psi\right|H_j\left|\psi\right\rangle^2+\sum_{j=1}^m\left\langle\psi\right|A_j\left|\psi\right\rangle^2.
		\end{equation}
        This means that the equivalence class remains the same upon replacing each Kraus operator $K_j$ with~a~pair of~operators $H_j$ and~$A_j$. 
        As a~result, for~the~purposes of~analyzing properties of the fidelity, it is sufficient to~only consider channels described by~Hermitian Kraus operators. 
        In~such a~case, the operation fidelity is given by the squared Euclidean length of~the~vector $\left\langle\psi\right|\vec{K}\left|\psi\right\rangle$. 
        Its minimum and~maximum are therefore equal to~the~squares of~the Crawford number and~the joint numerical radius of a collection of operators $\left(H_1,H_2,\ldots,H_m,A_1,A_2,\ldots,A_m\right)$, respectively.
        Finally, we are able to conclude the section with propositions that shall prove useful in subsequent investigations.
	
	\begin{proposition}
	For any operation $\Phi$ acting on a system of dimension $d$ and given by a set of $m$ Kraus operators $\{K_1,\ldots,K_m\}$, each with the Hermitian and anti-Hermitian parts $K_j=H_j +\mathrm{i} A_j$, the maximal $F_\text{max}$ and minimal $F_\text{min}$ operation fidelity can be written, respectively, as the squares of the joint numerical radius and the Crawford number:
	\begin{equation}
         F_\text{max}=\left[w\left(H_1,\ldots,H_m,A_1,\ldots,A_m\right)\right]^2,
	\end{equation}
	and
	\begin{equation}
         F_\text{min}=\left[c\left(H_1,\ldots,H_m,A_1,\ldots,A_m\right)\right]^2.
	\end{equation}
	\end{proposition}
	
	Analogously, using the numerical shadow we can evaluate the probability distribution of the operation fidelity.
	
	\begin{proposition}\label{prop:kraus_fidelity}
	For any operation $\Phi$ given by a set of Kraus operators $\{K_1,\ldots,K_m\}$, decomposed into their Hermitian and anti-Hermitian parts, $K_j=H_j +\mathrm{i} A_j$, probability density function of the operation fidelity $P\left(F\right)$ is given by 
	\begin{equation}\label{eq:fid_dist_from_nom_shad}
	    P(F) = \int_W P_{H_1,\ldots,H_m,A_1,\ldots,A_m}\left(\vec{r}\right)\delta(F-|\vec{r}|^2)\mathrm{d}^{2m}r,
	\end{equation}
	where $W=W\left(H_1,\ldots,H_m,A_1,\ldots,A_m\right)\subset\mathbb{R}^{2m}$ denotes the joint numerical range of $m$ Kraus operators, while $P_{H_1,\ldots,H_m,A_1,\ldots,A_m}\left(\vec{r}\right)$ is the joint numerical shadow supported in $W$.
	\end{proposition}
	
	Starting from Proposition~\ref{prop:kraus_fidelity}, we derive the probability distribution of the operation fidelity for unitary channels in an alternative way. 
    This constitutes another characterization of operation fidelity equivalence classes.
    
    \begin{proposition}\label{prop:unitary_decomposition}
    Operation fidelity distribution of a unitary channel\, $U$ of an arbitrary dimension $d$ is identical to the operation fidelity distribution of the channel defined by two Kraus operators $K_1 = \mathrm{Herm}\left(U\right) = \frac{U+U^\dagger}{2}$ and $K_2 = \mathrm{Anti}\left(U\right) = \frac{U-U^\dagger}{2\mathrm{i}}$.
    \end{proposition}
	
	These constructive results show that several properties of a quantum channel can be derived analyzing the joint numerical range of its Kraus operators.
	
	\subsection{Probabilistic estimation of channels}\label{subsection:probabilistic_estimation}
   Consider a family of channels $\{\Phi_\theta\}$. 
    Channel estimation, in its full generality, 
    consists in determination of the parameter $\theta$. 
    This notion is also known in the literature as channel identification or channel discrimination, if there are two channels to choose from~\cite{Sidhu_2020}.
    This problem is related to the quantum Fisher information, which quantifies the achievable precision of quantum channel estimation. 
    
    The simplest way to estimate a channel through quantum process tomography is not practical because of the enormous number of measurements required even in the simplest settings~\cite{DeMartini_2003}.
    Due to their usefulness in quantum metrology, a plethora of different variants of channel estimation have been discussed in the past 20 years~\cite{Sarovar_2006,Liang_2002,Kitaev_1996,Rudolph_2003,Fujiwara_2003,Fujiwara_2004,Sasaki_2002,Zheng_2014}. 
    Here, we wish to motivate yet unexplored setup that uses operational fidelity. 

    Suppose we are given an unknown channel $\Phi_\theta$ and a time-unvarying source of the same pure state $\ket{\psi}$. 
    Assume we do not know what is the exact state, as it is chosen at random according
    to the Haar measure on the corresponding Hilbert space. 
    Our goal is to estimate the parameter $\theta$, but since state $\ket{\psi}$ is beyond our control, usual channel estimation techniques are not helpful. 
    This leaves us with three basic possibilities: (a) quantum state tomography of both input and output states, (b) using the known experimental setup for finding the fidelity between input and output states~\cite{Hendrych_2003}, or (c) some more general scenario involving application of the channel in a sequential way. 
    
    Here, we shall restrict our investigation to option (b), since tomography forms an already widely-studied subfield.
    With this analysis, we do not aim to verify, whether the fidelity-based channel estimation is the most efficient one in the general case, which is outside of our primary inspection.
    Rather, we wish to introduce a novel framework that might be useful for channel estimation for specific laboratory setups.  

    As an example of the discrimination scenario, consider two channels: $\Phi$ and its noisy version $BF\circ\Phi$, where $BF$ corresponds to a bit flip, realized with Kraus operators $\{\sqrt{(1-p)}\,\mathbb{I}, \sqrt{p}\, \sigma_x\}$ with noise parameter $p$.
    Let $\Phi_{?}$ denote the one of these two channels we have at our disposal. 
    
    We start by using the source to produce the unknown pure state $\ket{\psi}$, then we wish to evaluate its fidelity $F_\psi = \bra{\psi}\Phi_{?}(\ket{\psi}\bra{\psi})\ket{\psi}$ with the final state after action of the channel. 
    To do so, we need several copies of both the initial and the final state~\cite{Hendrych_2003}; therefore, we apply the channel many times on different copies. 
    Having estimated the operation fidelity
    for random input states, we compare the densities of the probability of fidelity $P_{\Phi}(F_\psi)$ and $P_{BF\circ\Phi}(F_\psi)$ for a given value of $F_\psi$. 
    Finally, if $P_{\Phi}(F_\psi) > P_{BF\circ\Phi}(F_\psi)$ then we estimate the channel to be $\Phi$ -- in the opposite case, the estimator is $BF\circ\Phi$.
    
    To show the efficacy of this setup, suppose $\Phi$ is a Haar-random qubit unitary channel. 
    We assume that with probability 50\% the channel is ideal and with probability 50\% the channel is composed with bit-flip noise of error parameter $p$.
    Using the setup specified above, the probability of correct discrimination for a random qubit unitary channel can be evaluated to be
    \begin{equation}\label{eq:prob_discrimination}
        P_\mathrm{corr} = \bigg\langle\frac{1}{2} + \frac{1}{4}\int_0^1 |P_{U}(F) - P_{BF\circ U}(F)|dF \bigg\rangle_{U(2)},
    \end{equation}
    where the integral has an interpretation of the area of non-overlapping regions under the distribution curves. 
    The average over the unitary group of size 2 is denoted by the outermost $\langle \,\,\rangle_{U(2)}$.

    The average quality of channel discrimination for this task, depending on the noise parameter $p$, is depicted in Fig.~\ref{fig:discrimination}. 
    Starting from an expected 50\% chance for discrimination success between two equal channels \mbox{($p=0$)}, the probability of success settles a little above 75\% for $p\geq 0.4$. 
    Since the resources available for this task are different from the standard channel discrimination, it is not possible to compare the efficiencies.

 \begin{center}
	\begin{figure}[H]
	\centering
    \begin{tikzpicture}[scale=0.65, every node/.style={scale=0.63}]
        \newcommand{\deltax}{2.28} 
	\node (a) at (0,0) {\includegraphics[width=0.7\textwidth]{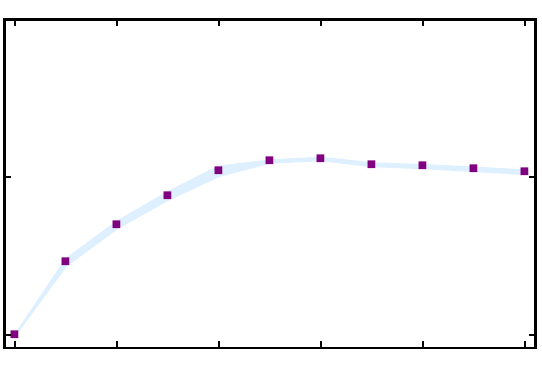}};
	\node (a) at (0,-4.4) {\scalebox{1.54}{$p$}};
        \node (a) at (-5.7,-4.1) {\scalebox{1.54}{$0$}};
	\node (a) at (-5.7+\deltax,-4.1) {\scalebox{1.54}{$0{.}2$}};
	\node (a) at (-5.7+2*\deltax,-4.1) {\scalebox{1.54}{$0{.}4$}};
        \node (a) at (-5.7+3*\deltax,-4.1) {\scalebox{1.54}{$0{.}6$}};
        \node (a) at (-5.7+4*\deltax,-4.1) {\scalebox{1.54}{$0{.}8$}};
	\node (a) at (-5.7+5*\deltax,-4.1) {\scalebox{1.54}{$1$}};
	\node (a) at (-6.5,1.8) {\scalebox{1.54}{\begin{turn}{90}$P_\mathrm{corr}$\end{turn}}};
	\node (a) at (-6.49,-3.4) {\scalebox{1.54}{$0.5$}};
	\node (a) at (-6.59,0.1) {\scalebox{1.54}{$0{.}75$}};
        \node (a) at (-6.3,3.6) {\scalebox{1.54}{$1$}};
    \end{tikzpicture}
    \caption{The average probability of correct discrimination $P_\mathrm{corr}$, given by Eq.~(\ref{eq:prob_discrimination}), between random qubit unitary channel $U$ and its noisy version $BF\circ U$, using the fidelity-based scenario explained in the text. 
    Horizontal axis represents the strength $p$ of the bit-flip noise, given by the Kraus operators of the bit-flip channel $\{\sqrt{(1-p)}\,\mathbb{I}, \sqrt{p}\, \sigma_x\}$.
    Light blue background denotes the standard error of the mean.}\label{fig:discrimination}
	\end{figure}
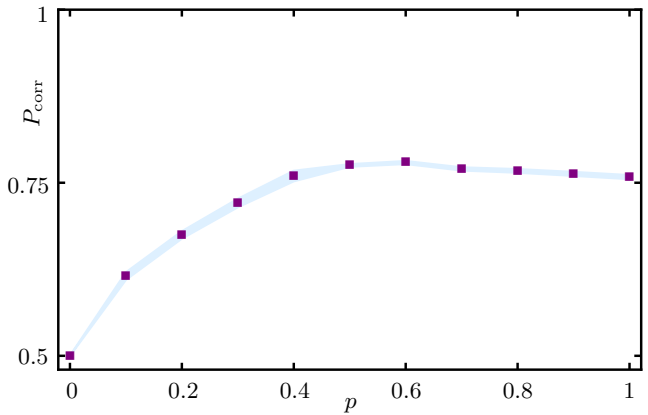
    \end{center}

    Similarly, one can analyze discrimination between any two channels, provided they belong to different equivalence classes, introduced in Subsection~\ref{subsec:applications_of_numerical_range}. 
    However, generic quantum channels will typically have different distributions of the operation fidelity: thus analyzing $P(F)$ can be decisive to demonstrate that these two channels belong to different equivalence classes. 
    It is also thinkable to expand this setup into a general channel estimation scheme, with more than two possible channels to choose from.
    
    Finally, we would like to convince the reader that experimental evaluation of fidelity between unknown states is 
    a feasible task.
    One example of~such an~experimental setup uses the following expression for fidelity between two states $\rho$ and~$\sigma$ on~a~Hilbert space $\mathcal{H}$, of~which at~least one is pure~\cite{Hendrych_2003},
    \begin{equation}
        F\left(\rho,\sigma\right)=\mathrm{Tr}\left(\Pi^+ \rho\otimes\sigma\right)-\mathrm{Tr}\left(\Pi^- \rho\otimes\sigma\right),
    \end{equation}
    where $\Pi^+$ and~$\Pi^-$ are projections onto symmetric and~anti-symmetric subspaces of~$\mathcal{H}\otimes\mathcal{H}$, respectively.
    Thus, in~order to~measure $F\left(\rho,\sigma\right)$, it is sufficient to~measure expectation values of~observables $\Pi^+$ and~$\Pi^-$ on~the~state $\rho\otimes\sigma$. If $\rho$ and~$\sigma$ are polarization states of photons, projections $\Pi^+$ and~$\Pi^-$ can be realized by~a~beam splitter, which can distinguish between symmetric and~anti-symmetric states of photon pairs. 
    In addition to this experimental setup, recently several algorithms for determining fidelity were introduced, using machine learning, hybrid algorithms, or shadow tomography~\cite{Wang_2023,Cerezo_2020,Zhang_2021,urRehman_2022,Flammia_2011,daSilva_2011,Huang_2020,Miszczak_2009,Guhne_2007}.

    To sum up, since it is possible to~experimentally measure fidelity between two quantum states, by experimental samples of \emph{state fidelity} between input and output states,  
    we are able to determine
    with a high probability, 
    to which equivalence class the channel belongs by the mean of the \emph{channel fidelity} distribution. 
    This exemplifies the usefulness of the study of the distribution of the operation fidelity for channel estimation.
    
\section{Probability distributions for~operation fidelity \\ (for \texorpdfstring{$d=2$}{Lg} and \texorpdfstring{$d=3$}{Lg})}\label{sec:PDF_for_qubit_qutrit}
    
	There are several ways to~obtain fidelity probability distributions for~quantum channels. For~qubit channels, it can be done by~parametrizing the~space of~pure states as $\left|\psi\right\rangle=\cos\frac{\theta}{2}\left|0\right\rangle+\cos\frac{\theta}{2}\mathrm{e}^{\mathrm{i}\phi}\left|1\right\rangle$ with probabilistic measure given by~$\frac{1}{4\pi}\sin\theta\mathrm{d}\theta\wedge\mathrm{d}\phi$, then changing variables so~that one of~them is operation fidelity and~integrating the~measure over the~second variable. 
	In~higher dimensions (for example, for~qutrits) this method may be difficult to~use (due to the more complicated geometry of~the~pure state space and~probabilistic measure), so taking advantage of the properties of the numerical shadow will prove useful.
 
    This simple method allows one to obtain statistical properties for single-qubit channels, as shown in Table~\ref{tab:special_qubit_channels} for three examples:
    a unitary channel, an exemplary phase damping channel whose Kraus operators are orthogonal projections and a channel defined by 2 nilpotent Kraus operators.
	The results were acquired by the numerical range method sketched above.
 
    \onecolumngrid    
    
	\begin{table*}[h]
	    \centering
        \begin{tabular}{|c|c|c|c|c|}
        \hline
         Kraus operators & Density $P(F)$ & Support of $P(F)$ & Average $\left\langle F\right\rangle$ & Standard deviation $\sqrt{\left\langle F^2\right\rangle-\left\langle F\right\rangle^2}$ \\ \hline
         a) $U=\ket{0}\!\bra{0}+\mathrm{e}^{\mathrm{i}\alpha}\ket{1}\!\bra{1}$ & $\left[\left(1-\cos\alpha\right)\left(2F-1-\cos\alpha\right)\right]^{-1/2}$ & $\left[\left(1+\cos\alpha\right)/2,1\right]$ & $\left(2+\cos\alpha\right)/3$ & $\left(2\sin^2\frac{\alpha}{2}\right)/\left(3\sqrt{5}\right)$ \\ \hline
         b) $K_1=\ket{0}\!\bra{0}$, $K_2=\ket{1}\!\bra{1}$ & 	$\left(2F-1\right)^{-1/2}$ & $\left[1/2,1\right]$ & $2/3$ & $1/\left(3\sqrt{5}\right)$ \\ \hline
         c) $K_1 = \ket{0}\!\bra{1}$, $K_2=\ket{1}\!\bra{0}$ & $\left(1-2F\right)^{-1/2}$ & $\left[0,1/2\right]$ & $1/3$ & $1/\left(3\sqrt{5}\right)$ \\ \hline
        \end{tabular}
        \caption{Statistical properties of the operation fidelity $F$ for exemplary single-qubit channels: a) unitary channel, b) orthogonal projection, and c) bit-flip channel. The probability density function of the operation fidelity is denoted by $P\left(F\right)$.}\label{tab:special_qubit_channels}
    \end{table*}
    \twocolumngrid

    Note that, in case a) for $\alpha=\frac{\pi}{2}$, the distribution of the operation fidelity is the same as in case b). This follows from the fact that in occasion a) matrix $U=\ket{0}\!\bra{0}+\mathrm{e}^{\frac{\mathrm{i}\pi}{2}}\ket{1}\!\bra{1}=\ket{0}\!\bra{0}+\mathrm{i}\ket{1}\!\bra{1}=K_1+\mathrm{i}K_2$, where $K_1$ and $K_2$ are the Kraus operators from instance b). Proposition~\ref{prop:unitary_decomposition} implies that for $\alpha=\frac{\pi}{2}$ the quantum channels from cases a) and b) have identical distributions of operation fidelity. 
    Note that, although the channels are different, 
    they belong to the same equivalence class defined in Subsection~\ref{subsec:applications_of_numerical_range}.
	
	In the following subsections, we shall consider the study of qubit and qutrit channels, $\rho \mapsto \sum_j K_j \rho K^\dagger_j$, given by Kraus operators $K_j$.
	The method presented here works for $d=2$. For larger dimensions one needs to use  the general expression (\ref{eq:fid_dist_from_nom_shad}) and integrate over the joint numerical range $W\left(K_1,\ldots,K_m\right)$.

\subsection{Mixed unitary channels acting on a single qubit system}\label{sec:mixed_unitary}
	A~mixed unitary channel is~a~quantum channel described by~$m$ Kraus operators~(\ref{eq:definition_Kraus}), each proportional to~a~unitary operator:
        \begin{equation}
            \Phi\left(\rho\right)=\sum_{j=1}^m p_j U_j \rho U_j^\dagger,
        \end{equation}
    where $p_j$ are non-negative real numbers satisfying $\sum_{j=1}^m p_j=1$. 
    Using the Pauli matrices $\vec{\sigma}$, we can parametrize the unitary operators as rotations around 3-dimensional normalized vectors $\vec{n}_j$ by angles $\theta_j$
        \begin{equation}
            U_j=\mathrm{e}^{\mathrm{i}\frac{\theta_j}{2}\vec{n}_j\cdot\vec{\sigma}}=\cos\frac{\theta_j}{2}\mathds{1}+\mathrm{i}\sin\frac{\theta_j}{2}\vec{n}_j\cdot\vec{\sigma}.
        \end{equation}
    
    
    The operation fidelity $F$ of channel $\Phi$ for a pure state with a Bloch vector $\vec{b}$ can be expressed through Eq.~(\ref{eq:fidelity_Kraus}),
    \begin{equation}\label{eq:qubit_mixed_unitary_fidelity}
        F=\sum_{j=1}^m p_j\left(\cos^2\frac{\theta_j}{2}+\sin^2\frac{\theta_j}{2}\left(\vec{n}_j\cdot\vec{b}\right)^2\right) =\vec{b}^TS \, \vec{b},
    \end{equation}
    where $S$ is a real symmetric matrix of order $3$ defined as
    \begin{equation}
        S_{kl}=\sum_{j=1}^m p_j\left(\cos^2\frac{\theta_j}{2}\delta_{kl}+\sin^2\frac{\theta_j}{2}\left(\vec{n}_j\right)_k\left(\vec{n}_j\right)_l\right).
    \end{equation}
    
    Since $F\geq 0$, then $S$ has three real non-negative eigenvalues.
    Using Eq.~(\ref{eq:qubit_mixed_unitary_fidelity}), we conclude that the extremal values of the operation fidelity for~this~channel are equal to the extremal eigenvalues of~matrix $S$, or, equivalently, to the Crawford number and the numerical radius,
        \begin{equation}
        \begin{split}
            F_{\min}&=\min_{\ket{\psi}}\braket{\psi|\Phi\left(\ket{\psi}\bra{\psi}\right)|\psi}=c(S),
            \\
            F_{\max}&=\max_{\ket{\psi}}\braket{\psi|\Phi\left(\ket{\psi}\bra{\psi}\right)|\psi}=w(S).
        \end{split}
    \end{equation}

    The average operation fidelity can be obtained integrating Eq.~(\ref{eq:qubit_mixed_unitary_fidelity}) over the Bloch sphere 
    \begin{equation}
        \left\langle F\right\rangle=\sum_{k=1}^m\frac{p_k\left(2\cos^2\alpha_k+1\right)}{3}.
    \end{equation}
    This result agrees with~the~known formula connecting the average fidelity of~a~quantum channel with traces of~its Kraus operators~\cite{Nielsen_2002}.

    To derive the distribution of the operation fidelity, we use a basis in which matrix $S$ is diagonal, $S =\mathrm{diag}\left(\lambda_1,\lambda_2,\lambda_3\right)$, where $0\leq\lambda_1\leq\lambda_2\leq\lambda_3$. Using the~standard parametrization of the Bloch sphere, $\ket{\psi}=\cos\frac{\theta}{2}\ket{0}+\mathrm{e}^{\mathrm{i}\varphi}\sin\frac{\theta}{2}\ket{1}$, we integrate over $\theta$ to obtain an~analytical expression for the operation fidelity distribution:
    \begin{widetext}
    \begin{equation}\label{eq:mixed_unitary_dist}
        P\left(F\right)=\frac{1}{4\pi}\int_0^{2\pi}\frac{\mathrm{d}\varphi}{\sqrt{\left(F-\lambda_1\cos^2\varphi-\lambda_2\sin^2\varphi\right)\left(\lambda_3-\lambda_1\cos^2\varphi - \lambda_2\sin^2\varphi\right)}}.
    \end{equation}    
    \end{widetext}
    This formula is valid if~$\lambda_2<F\leq\lambda_3$. A suitable formula for the case $\lambda_1\leq F<\lambda_2$ can be obtained by permuting $\lambda_1$ with $\lambda_3$.
    For~$F\rightarrow\lambda_2$, the density $P\left(F\right)$ becomes infinite.

	\subsubsection{Pauli channels acting on~qubits}
		Pauli channels form a subset of mixed unitary qubit channels with Kraus operators proportional to Pauli matrices $\sigma_j$ and the identity matrix. They can be written as
		\begin{equation}
		    \Phi\left(\rho\right)=p_0 \rho+\sum_{j=1}^3 p_j \sigma_j \rho \sigma_j,
		\end{equation}
		where $p_0+p_1+p_2+p_3=1$. Using Eq.~(\ref{eq:mixed_unitary_dist}), we can express the operation fidelity distribution of the channel as
		\begin{widetext}
		\begin{equation}\label{eq:pauli_channels}
			P\left(F\right)=\frac{1}{4\pi}\int_0^{2\pi}\frac{\mathrm{d}\varphi}{\sqrt{\left(F-p_0-p_1\cos^2\varphi-p_2\sin^2\varphi\right)\left(p_3-p_1\cos^2\varphi-p_2\sin^2\varphi\right)}}.
		\end{equation}
		\end{widetext}
		If~$p_1=p_2$, this formula simplifies to
		\begin{equation}
			P\left(F\right)=\frac{1}{2\sqrt{\left(F-p_0-p_1\right)\left(p_3-p_1\right)}}.
		\end{equation}
		Furthermore, if $p_1=0$, this formula is reduced to the distribution of the operation fidelity obtained for case a) from Table~\ref{tab:special_qubit_channels} -- a unitary channel defined by the operator $U=\ket{0}\!\bra{0}+\mathrm{e}^{\mathrm{i}\alpha}\ket{1}\!\bra{1}$ with $p_0=\cos^2\frac{\alpha}{2}$ and $p_3=\sin^2\frac{\alpha}{2}$.
	\subsection{Unitary channels acting on qutrits}
	    In this subsection, we shall focus on the study of a single unitary operator acting on qutrits. 
	    Any unitary operator is diagonalizable; therefore, we will solely concentrate on a study of statistical properties of fidelities of diagonal channels since the probability distributions depend only on the eigenvalues of a given unitary operator.
	    
		Here, we consider a unitary operator $U$ of order $3$.
		Since the overall phase of $U$ does not influence its operation fidelity and cannot be physically measured, we can assume that one of its eigenvalues is $1$. 
		Upon diagonalization, this operator can then be expressed as $U = \left|0\right\rangle\!\left\langle 0\right|+\mathrm{e}^{\mathrm{i}\alpha}\left|1\right\rangle\!\left\langle 1\right|+\mathrm{e}^{\mathrm{i}\beta}\left|2\right\rangle\!\left\langle 2\right|$, where $\alpha,\beta\in\left[0,2\pi\right]$ and $\alpha\leq\beta$. 
		
    Thus, for any unitary operator $U$ of size $d=3$ the numerical range $W$ forms a triangle on the complex plane with vertices: $1$, $\mathrm{e}^{\mathrm{i}\alpha}$, and $\mathrm{e}^{\mathrm{i}\beta}$, while its numerical shadow $P_U\left(z\right)$ is uniform on the said triangle~\cite{Kippenhahn_1951}. 
		Hence, to obtain the probability distribution of operation fidelity, it suffices to integrate the uniform shadow $P_U\left(z\right)$ over the phase $\varphi$, where $z=\sqrt{F}\mathrm{e}^{\mathrm{i}\varphi}$, as given by Eq.~(\ref{eq:fid_dist_from_nom_shad}).
		Due to the uniformity of the numerical shadow this integration can be performed analytically, resulting in the following, compact expression for the density of the operation fidelity, where $\Re$ denotes the real part of a potentially complex value of the $\arccos$ function:
		\begin{widetext}
		\begin{equation}
			P\left(F\right)=\frac{2\Re\left(\arccos\left(\frac{\cos\left(\frac{\beta}{2}\right)}{\sqrt{F}}\right)-\arccos\left(\frac{\cos\left(\frac{\alpha}{2}\right)}{\sqrt{F}}\right)-\arccos\left(\frac{\cos\left(\frac{\beta-\alpha}{2}\right)}{\sqrt{F}}\right)\right)}{\sin\left(\alpha\right)+\sin\left(\beta-\alpha\right)-\sin\left(\beta\right)}.
			\label{unitary_qutrits}
		\end{equation}		
		\end{widetext}
		
		In its full generality (for $d=3$), the distribution of the operation fidelity $P\left(F\right)$ has three cusps -- see Fig.~\ref{fig:triangle_radii}. 
		Their nature can be explained using the numerical range of the unitary matrix $U$ -- each of the cusps represents a change of the integration domain (see Fig.~\ref{fig:triangle_radii}) and corresponds to a value of the operation 
  fidelity $F_i$ for which a circle on complex plane with center $0$ and radius $r_i=\sqrt{F_i}$ is tangent to at least one side of the numerical range (which forms a triangle). 
		These critical values read $F_1=\cos^2\left(\frac{\alpha}{2}\right)$, $F_2=\cos^2\left(\frac{\beta}{2}\right)$ and $F_3=\cos^2\left(\frac{\beta-\alpha}{2}\right)$.
  
    \onecolumngrid
    
    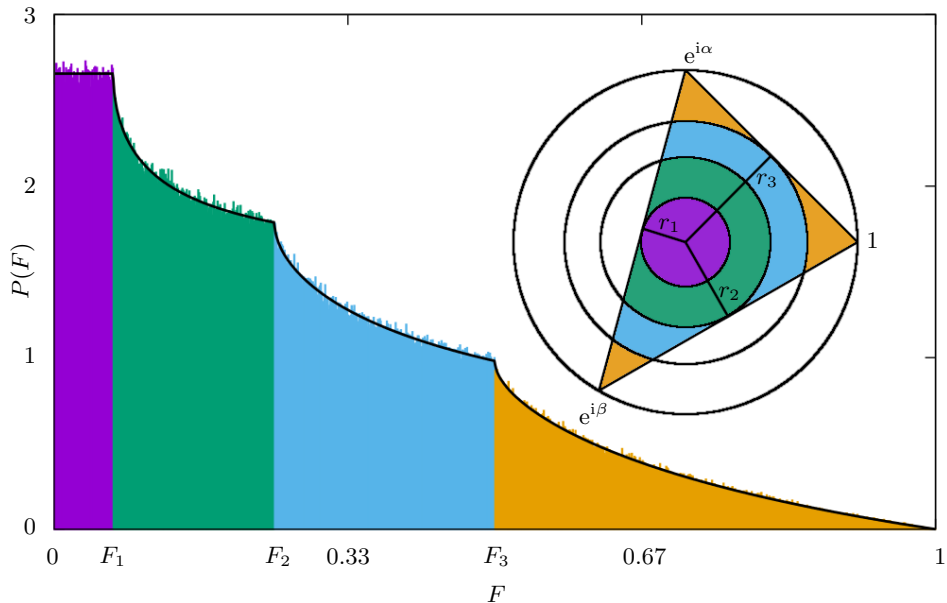
\begin{figure}[H]
	\centering
	\begin{tikzpicture}
	\node (a) at (0,0) {\includegraphics[width=0.7\textwidth]{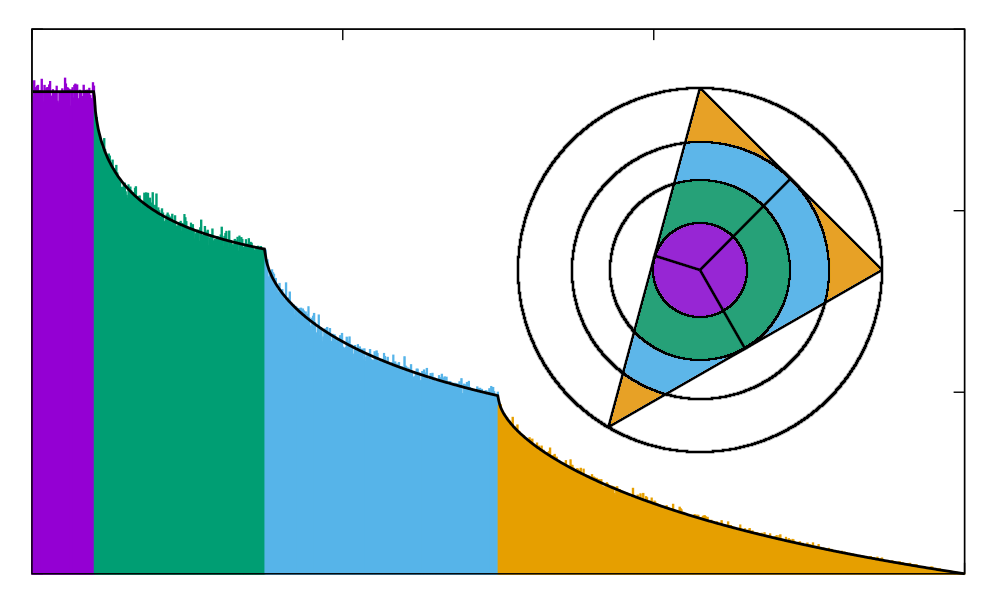}};
	\node (a) at (0,-4.3) {$F$};
    \node (a) at (-5.9,-3.8) {$0$};
	\node (a) at (-1.97,-3.8) {$0{.}33$};
	\node (a) at (1.97,-3.8) {$0{.}67$};
	\node (a) at (5.9,-3.8) {$1$};
	\node (a) at (-6.3,0) {\begin{turn}{90}$P(F)$\end{turn}};
	\node (a) at (-6.2,-3.4) {$0$};
	\node (a) at (-6.2,-1.13) {$1$};
	\node (a) at (-6.2,1.13) {$2$};
    \node (a) at (-6.2,3.4) {$3$};
    \node (a) at (5,0.4) {$1$};
    \node (a) at (2.7,2.9) {$\mathrm{e}^{\mathrm{i}\alpha}$};
    \node (a) at (1.3,-1.9) {$\mathrm{e}^{\mathrm{i}\beta}$};
    \node (a) at (2.3,0.6) {$r_1$};
    \node (a) at (3.1,-0.3) {$r_2$};
    \node (a) at (3.6,1.2) {$r_3$};
    \node (a) at (-5.1,-3.8) {$F_1$};
    \node (a) at (-2.9,-3.8) {$F_2$};
    \node (a) at (0,-3.8) {$F_3$};
    \end{tikzpicture}
    \caption{
    Distribution $P\left(F\right)$ of operation fidelity for~a~qutrit channel defined by a unitary operator $U = \text{diag}(1,\mathrm{e}^{\mathrm{i}\alpha},\mathrm{e}^{\mathrm{i}\beta})$ with $\alpha=\frac{\pi}{2}$ and $\beta=\frac{4\pi}{3}$: numerical data compared with analytical curve obtained by Eq.~(\ref{unitary_qutrits}).
    Distribution $P(F)$ can be obtained by integrating over the numerical range: $W(U) = \triangle (1,e^{i\alpha}, e^{i\beta})$.
	This triangle can be split into fragments of rings with radii $r_1$, $r_2$ and $r_3$ --- see the inset.
    The cusps of the graph occur for values of the operation fidelity equal to squares of the radii $r_1$, $r_2$ and $r_3$ for which integration domain changes: $F_1=r_1^2=\cos^2\left(\frac{\beta-\alpha}{2}\right)=\left(2-\sqrt{3}\right)/4$, $F_2=r_2^2=\cos^2\left(\frac{\beta}{2}\right)=1/4$ and $F_3=r_3^2=\cos^2\left(\frac{\alpha}{2}\right)=1/2$.}\label{fig:triangle_radii}
	\end{figure}
 
    \twocolumngrid
	
	Examples of~distributions obtained using Eq.~(\ref{unitary_qutrits}) for~different channels are shown in Fig.~\ref{fig:examples_qutrit_PDF}.
	
    \begin{figure*}
        \begin{tikzpicture}[scale=0.9, every node/.style={scale=0.9}]
    \node (a) at (0,0) {\includegraphics[width=0.48\textwidth]{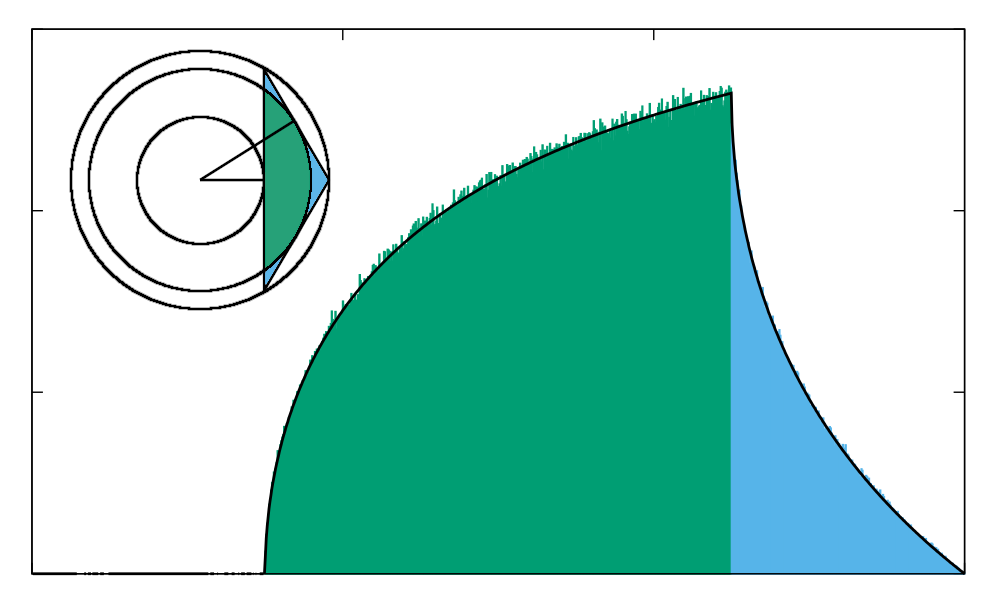}};
    \node (a) at (0,+2.8) {(a) $\alpha=\frac{\pi}{3}$, $\beta=\frac{5\pi}{3}$};
    \node (a) at (0,-3.1) {$F$};
    \node (a) at (-4,-2.7) {$0$};
    \node (a) at (-1.33,-2.7) {$0{.}33$};
    \node (a) at (1.33,-2.7) {$0{.}67$};
    \node (a) at (4,-2.7) {$1$};
    \node (a) at (-4.5,0) {\begin{turn}{90}$P(F)$\end{turn}};
    \node (a) at (-4.3,-2.3) {$0$};
    \node (a) at (-4.5,-0.77) {$0{.}83$};
    \node (a) at (-4.5,0.77) {$1{.}67$};
    \node (a) at (-4.4,2.3) {$2{.}5$};
    \node (a) at (-1.3,1) {$1$};
    \node (a) at (-1.7,2.1) {$\mathrm{e}^{\mathrm{i}\alpha}$};
    \node (a) at (-1.9,-0.1) {$\mathrm{e}^{\mathrm{i}\beta}$};
    \node (a) at (-2.3,0.9) {$r_1$};
    \node (a) at (-1.8,1.2) {$r_2$};
    \node (a) at (-2,-2.7) {$F_1$};
    \node (a) at (2,-2.7) {$F_2$};
\end{tikzpicture}
        \hspace*{\fill}
        \begin{tikzpicture}[scale=0.9, every node/.style={scale=0.9}]
    \node (a) at (0,0) {\includegraphics[width=0.48\textwidth]{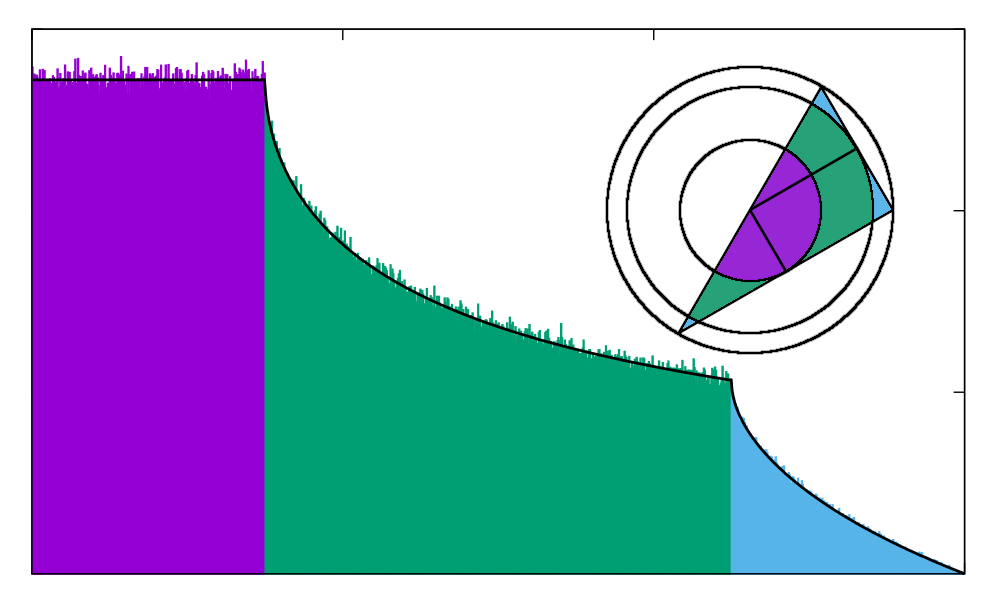}};
    \node (a) at (0,+2.8) {(b) $\alpha=\frac{\pi}{3}$, $\beta=\frac{4\pi}{3}$};
    \node (a) at (0,-3.1) {$F$};
    \node (a) at (-4,-2.7) {$0$};
    \node (a) at (-1.33,-2.7) {$0{.}33$};
    \node (a) at (1.33,-2.7) {$0{.}67$};
    \node (a) at (4,-2.7) {$1$};
    \node (a) at (-4.5,0) {\begin{turn}{90}$P(F)$\end{turn}};
    \node (a) at (-4.3,-2.3) {$0$};
    \node (a) at (-4.5,-0.77) {$0{.}67$};
    \node (a) at (-4.5,0.77) {$1{.}33$};
    \node (a) at (-4.3,2.3) {$2$};
    \node (a) at (3.6,0.8) {$1$};
    \node (a) at (3,2) {$\mathrm{e}^{\mathrm{i}\alpha}$};
    \node (a) at (1.3,-0.3) {$\mathrm{e}^{\mathrm{i}\beta}$};
    \node (a) at (2.5,0.5) {$r_1$};
    \node (a) at (2.9,1) {$r_2$};
    \node (a) at (-2,-2.7) {$F_1$};
    \node (a) at (2,-2.7) {$F_2$};
\end{tikzpicture}
        
        \begin{tikzpicture}[scale=0.9, every node/.style={scale=0.9}]
    \node (a) at (0,0) {\includegraphics[width=0.48\textwidth]{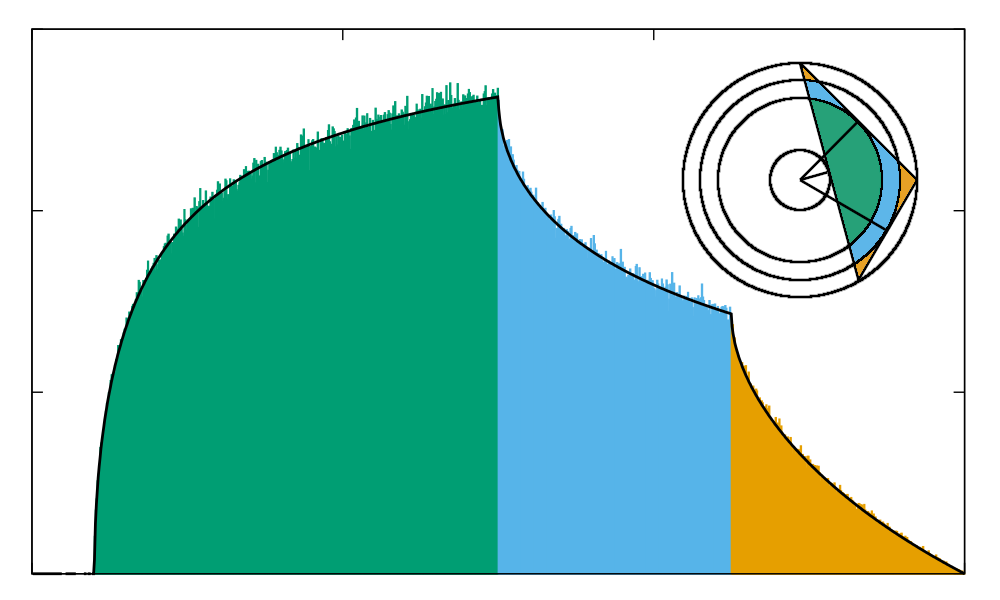}};
    \node (a) at (0,+2.8) {(c) $\alpha=\frac{\pi}{2}$, $\beta=\frac{5\pi}{3}$}; 
    \node (a) at (0,-3.1) {$F$};
    \node (a) at (-4,-2.7) {$0$};
    \node (a) at (-1.33,-2.7) {$0{.}33$};
    \node (a) at (1.33,-2.7) {$0{.}67$};
    \node (a) at (4,-2.7) {$1$};
    \node (a) at (-4.5,0) {\begin{turn}{90}$P(F)$\end{turn}};
    \node (a) at (-4.3,-2.3) {$0$};
    \node (a) at (-4.5,-0.77) {$0{.}67$};
    \node (a) at (-4.5,0.77) {$1{.}33$};
    \node (a) at (-4.3,2.3) {$2$};
    \node (a) at (3.8,1) {$1$};
    \node (a) at (2.2,2.1) {$\mathrm{e}^{\mathrm{i}\alpha}$};
    \node (a) at (3.1,-0.1) {$\mathrm{e}^{\mathrm{i}\beta}$};
    \node (a) at (3,1.1) {$r_1$};
    \node (a) at (3.1,1.2) {$r_2$};
    \node (a) at (3.1,0.9) {$r_3$};
    \node (a) at (-3.4,-2.7) {$F_1$};
    \node (a) at (0,-2.7) {$F_2$};
    \node (a) at (2,-2.7) {$F_3$};
\end{tikzpicture}
        \hspace*{\fill}
        \begin{tikzpicture}[scale=0.9, every node/.style={scale=0.9}]
    \node (a) at (0,0) {\includegraphics[width=0.48\textwidth]{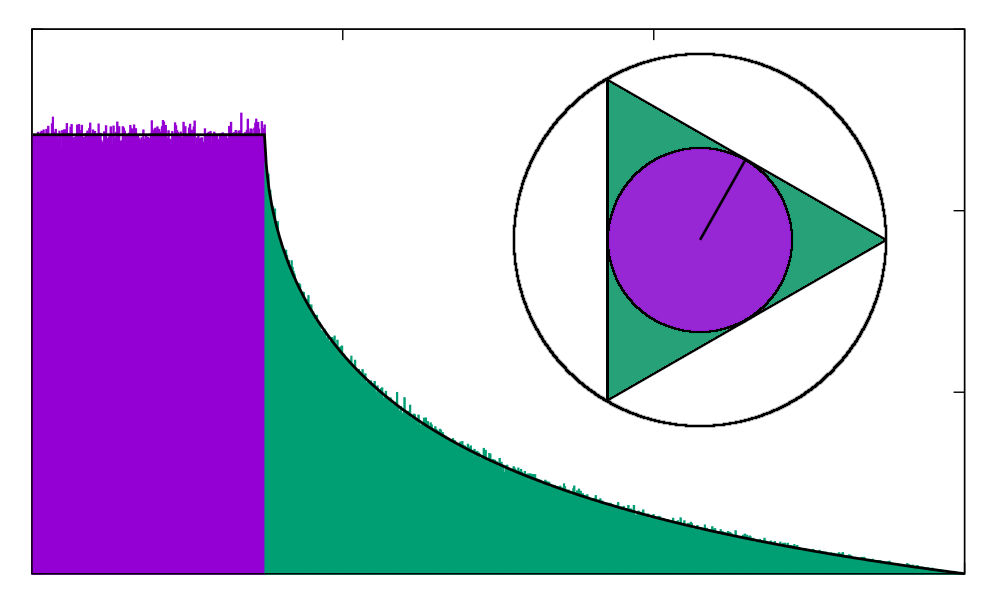}};
    \node (a) at (0,+2.8) {(d) $\alpha=\frac{2\pi}{3}$, $\beta=\frac{4\pi}{3}$};
    \node (a) at (0,-3.1) {$F$};
    \node (a) at (-4,-2.7) {$0$};
    \node (a) at (-1.33,-2.7) {$0{.}33$};
    \node (a) at (1.33,-2.7) {$0{.}67$};
    \node (a) at (4,-2.7) {$1$};
    \node (a) at (-4.5,0) {\begin{turn}{90}$P(F)$\end{turn}};
    \node (a) at (-4.3,-2.3) {$0$};
    \node (a) at (-4.3,-0.77) {$1$};
    \node (a) at (-4.3,0.77) {$2$};
    \node (a) at (-4.3,2.3) {$3$};
    \node (a) at (3.5,0.5) {$1$};
    \node (a) at (0.8,2) {$\mathrm{e}^{\mathrm{i}\alpha}$};
    \node (a) at (0.7,-1) {$\mathrm{e}^{\mathrm{i}\beta}$};
    \node (a) at (2.1,0.7) {$r_1$};
    \node (a) at (-2,-2.7) {$F_1$};
\end{tikzpicture}
        \caption{Distributions of the operation fidelity $P(F)$ for exemplary unitary qutrit channels defined by operators of the form $\mathrm{diag}\left(1,\mathrm{e}^{\mathrm{i}\alpha},\mathrm{e}^{\mathrm{i}\beta}\right)$ for different values of $\alpha$ and $\beta$.
        Solid curves denote distributions given by Eq.~(\ref{unitary_qutrits}), while underlying histograms refer to numerical data that fits perfectly analytical results.
        The existence of cusps is explained by the change of the integration domain as in Fig.~\ref{fig:triangle_radii} -- to highlight them we used different colors of histograms.
        The radii of circles tangent to sides of the triangle forming the numerical range $W$ (denoted by $r_i$) are connected with the cusps of the distribution $P\left(F\right)$ by $F_i = r_i^2$.}
        \label{fig:examples_qutrit_PDF}
    \end{figure*}
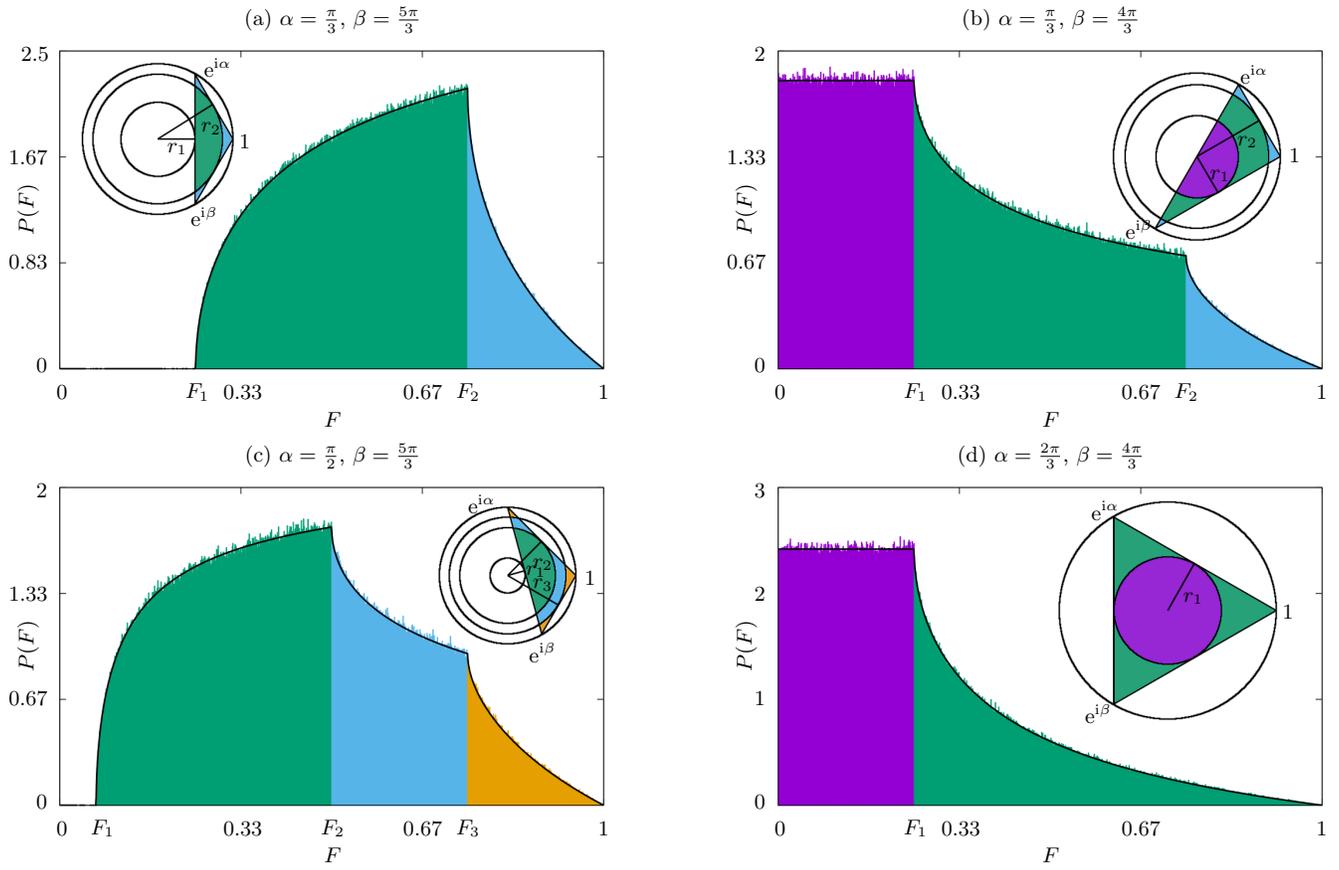           
        
\section{Minimal operation fidelity for Schur channels}\label{sec:diagonal_Kraus}
    This section is devoted to the study of a channel $\Phi$, acting on a system of an arbitrary dimension $d$, described by $m$ diagonal Kraus operators.
    Such channels are sometimes called {\sl Schur channels}.
    This is because every Schur-product channel that acts as $\Phi(\rho) = M \circ \rho$, where $\circ$ denotes the Schur (Hadamard) product, can be decomposed into diagonal Kraus operators~\cite{Bengtsson_Zyczkowski_geometry,Girard_2022}.

    Aside from~being mathematically convenient to~analyze, Schur channels also have a~certain physical significance that justifies further investigation of~their properties. Namely, such channels can be used to~describe a~phase flip noise -- a~noise which only changes the~relative phases of~probability amplitudes (in~standard computational basis). They can also be used to~model phase-damping (or~dephasing) channels \cite{Arqand_2020,Rexiti_2022}, which describe the~wash-out of~coherence and~play an~important role in~modern quantum technologies.
    
	Let~$\Phi$ be a~quantum Schur channel acting on~$d$-dimensional quantum states as
		\begin{equation}
			\Phi\left(\rho\right)=\sum_{j=1}^m K_j\rho K^\dagger_j,
		\end{equation}
		where $K_j$ denote $m$ Kraus operators that are all diagonal in~a~certain basis:
		\begin{equation}
K_j\left|i\right\rangle = \lambda_{ji}\left|i\right\rangle,
		\end{equation}
    with $i=1,\dots, m$ and $j=0,\dots, d-1.$
    Since $\Phi$ preserves the~trace of~density matrices, these eigenvalues satisfy
		\begin{equation}
     \sum_{j=1}^m\left|\lambda_{ji}\right|^2=1.
		\end{equation}

		Let $\left|\psi\right\rangle$ be a~pure state expressed in the computational basis $\left|i\right\rangle$ as
		\begin{equation}\label{eq:def_psi}
			\left|\psi\right\rangle=\sum_{i=0}^{d-1}\psi_i\left|i\right\rangle,
		\end{equation}
where $\psi_i$ are complex numbers satisfying $\sum_{i=0}^{d-1}\left|\psi_i\right|^2=1$.
The operation fidelity of~channel $\Phi$ acting on~state $\left|\psi\right\rangle$ is~then given by
    \begin{widetext}
    \begin{equation}
    \begin{split}
		F\left(\Phi,\left|\psi\right\rangle\right) &= \sum_{j=1}^m\left|\left\langle\psi\right|K_j\left|\psi\right\rangle\right|^2 = \sum_{j=1}^m\left|\sum_{i=0}^{d-1}\lambda_{ji}\left|\psi_i\right|^2\right|^2 =\sum_{j=1}^m\left(\sum_{i=0}^{d-1}\lambda^\star_{ji}\left|\psi_i\right|^2\right)\left(\sum_{l=0}^{d-1}\lambda_{jl}\left|\psi_l\right|^2\right)  = \\ &= \sum_{i,l=0}^{d-1}\left(\sum_{j=1}^m\lambda^\star_{ji}\lambda_{jl}\right)\left|\psi_i\right|^2\left|\psi_l\right|^2 = \sum_{i,l=0}^{d-1}\Re\left(\sum_{j=1}^m\lambda^\star_{ji}\lambda_{jl}\right)\left|\psi_i\right|^2\left|\psi_l\right|^2=\sum_{i,l=0}^{d-1}G_{il}p_i p_l,
		\label{fid_quad}
    \end{split}
    \end{equation}
where $G_{il}=\Re\left(\sum_{j=1}^m\lambda^\star_{ji}\lambda_{jl}\right)$ and~$p_i=\left|\psi_i\right|^2$. 
The matrix $G_{il}$ can be expressed as
    \begin{equation}
        G_{il}=\Re\left(\sum_{j=1}^m\lambda^\star_{ji}\lambda_{jl}\right)=\sum_{j=1}^m\left(\Re\left(\lambda^\star_{ji}\right)\Re\left(\lambda_{jl}\right)-\Im\left(\lambda^\star_{ji}\right)\Im\left(\lambda_{jl}\right)\right)=\sum_{j=1}^m\left(\Re\left(\lambda_{ji}\right)\Re\left(\lambda_{jl}\right)+\Im\left(\lambda_{ji}\right)\Im\left(\lambda_{jl}\right)\right),
    \end{equation}
    \end{widetext}
therefore it is the Gram matrix of a set of real vectors $\vec{v}_i=\left(\Re\left(\lambda_{1i}\right),\Re\left(\lambda_{2i}\right),\ldots,\Re\left(\lambda_{mi}\right),\Im\left(\lambda_{1i}\right),\Im\left(\lambda_{2i}\right),\ldots,\Im\left(\lambda_{mi}\right)\right)$ for $i\in\left\{0,1,\ldots,d-1\right\}$. If these vectors are linearly independent (which is only possible if $d\leq 2m$), this matrix is invertible.

	In~order to~find extrema of the operation fidelity we need to~extremize the~quadratic form given by Eq.~(\ref{fid_quad}) over the simplex defined by~
		\begin{align}
			\sum_{i=0}^{d-1}p_i&=1,\label{simp} 
		\end{align}
  with $p_i\ge0$ for $i=0,1, \ldots,d-1$.
Making use of the~Lagrange multipliers with constant $\mu$, we obtain
		\begin{equation}
   2\sum_{j=0}^{d-1}G_{ij}p_j-\mu=0,
			\label{cond}
		\end{equation}
for all $i\in \{0,...,d-1\}$. Multiplying this equation by~$p_i$ and~summing over~$i$ leads then to
		\begin{equation}
			\mu=2\sum_{i,j=0}^{d-1}G_{ij}p_i p_j=2F\left(\Phi,\left|\psi\right\rangle\right).
		\end{equation}
Simultaneously, Eq.~(\ref{cond}) leads to~the~following solution for~$p_i$:
		\begin{equation}
			p_i=\frac{\mu}{2}\sum_{j=0}^{d-1}\left(G^{-1}\right)_{ij}.
		\end{equation}
The normalization condition (\ref{simp}) gives rise to
		\begin{equation}
			\mu=\frac{2}{\sum_{i,j=0}^{d-1}\left(G^{-1}\right)_{ij}}.
		\end{equation}

Finally, the operation fidelity extremized over~the~interior of the simplex 
(\ref{simp}) reads
	\begin{equation}\label{eq:extremal_fidelity}
		F\left(\Phi,\left|\psi\right\rangle\right)_{\mathrm{extr}} = F\left(\Phi,\left|\psi\right\rangle\right)_{\mathrm{min}} = \frac{1}{\sum_{i,j=0}^{d-1}\left(G^{-1}\right)_{ij}}
	\end{equation}
and~is achieved for any state $\ket{\psi}$ whose coefficients (\ref{eq:def_psi}) satisfy
		\begin{equation}\label{eq:psi_non-negative}
			\left|\psi_i\right|^2=\frac{\sum_{j=0}^{d-1}\left(G^{-1}\right)_{ij}}{\sum_{k,l=0}^{d-1}\left(G^{-1}\right)_{kl}},
		\end{equation}
    provided that the right-hand side of Eq.~(\ref{eq:psi_non-negative}) is non-negative for all $i$. 
    Observe that Eq.~(\ref{eq:extremal_fidelity}) is valid if the global minimum of the quadratic form~(\ref{fid_quad}) lies inside the simplex~(\ref{simp}). 

        \begin{center}
	\begin{figure}[H]
	\centering
    \begin{tikzpicture}[scale=0.65, every node/.style={scale=0.65}]
	\node (a) at (0,0) {\includegraphics[width=0.7\textwidth]{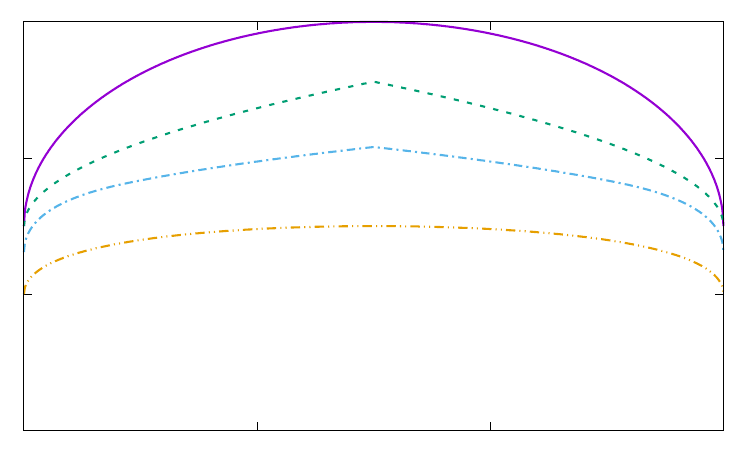}};
	\node (a) at (0,-4.2) {\scalebox{1.54}{$p$}};
    \node (a) at (-5.9,-3.8) {\scalebox{1.54}{$0$}};
	\node (a) at (-1.97,-3.8) {\scalebox{1.54}{$0{.}33$}};
	\node (a) at (1.97,-3.8) {\scalebox{1.54}{$0{.}67$}};
	\node (a) at (5.9,-3.8) {\scalebox{1.54}{$1$}};
	\node (a) at (-6.6,0) {\scalebox{1.54}{\begin{turn}{90}$F_{\text{min}}$\end{turn}}};
	\node (a) at (-6.2,-3.4) {\scalebox{1.54}{$0$}};
	\node (a) at (-6.5,-1.13) {\scalebox{1.54}{$0{.}33$}};
	\node (a) at (-6.5,1.13) {\scalebox{1.54}{$0{.}67$}};
    \node (a) at (-6.2,3.4) {\scalebox{1.54}{$1$}};
\end{tikzpicture}
    \caption{Minimal operation fidelity for~a~family of qubit diagonal channels $\Phi_p$ defined by two Kraus operators \mbox{$K_1 = \sqrt{p}\ket{0}\!\bra{0} + \sqrt{1-p}\ket{1}\!\bra{1}$} and $K_2 = \sqrt{1-p}\ket{0}\!\bra{0} + \sqrt{p}\ket{1}\!\bra{1}$, as well as qutrit diagonal channels $\Phi'_{p,q}$ defined by three Kraus operators \mbox{$K'_1 = \sqrt{p}\ket{0}\!\bra{0} + \sqrt{1-p}\ket{1}\!\bra{1}+\sqrt{q}\ket{2}\!\bra{2}$}, $K'_2 = \sqrt{1-p}\ket{0}\!\bra{0} + \sqrt{p}\ket{1}\!\bra{1}$ and $K'_3=\sqrt{1-q}\ket{2}\!\bra{2}$, with parameters $p$ and $q$ varying from 0 to 1. Solid purple line denotes the result for $\Phi_p$, while dashed green line, dash-dotted blue line and dash-double-dotted yellow line denote results for $\Phi'_{p,q}$, respectively for $q=1$, $q=0{.}3$ and $q=0$.
    Note that $\Phi_{1/2}$ is an identity channel; therefore, its operation fidelity $F=1$.}\label{fig:min_fid_diagonal_channels}
	\end{figure}
    \end{center}

    If this condition is not satisfied, the minimal operation fidelity corresponds to the boundary of simplex.
    It can be found as $F\left(\Phi,\left|\psi\right\rangle\right)_{\mathrm{min}} = \sum_{ij}G_{ij} p_i p_j$ for a vector $\vec{p} = \{p_i\}_{i=0}^{d-1}$ belonging to the simplex, such that the length of $\vec{p}$ in the metric $G_{ij}$ is minimal. 
    Similarly, the maximum operation fidelity always corresponds to a vector $\vec{p}$ on the boundary of the simplex of the maximal length in the metric $G_{ij}$.
    This is because the form~(\ref{fid_quad}) does not admit a global maximum.
    Nonetheless, these two extremal values do not admit such simplification as in Eq.~(\ref{eq:extremal_fidelity}). 
    In Fig.~\ref{fig:min_fid_diagonal_channels} we show the dependence of the minimum of the operation fidelity for a one-parameter family of diagonal qubit channels, for which the global minimum of the quadratic form~(\ref{fid_quad}) lies inside the simplex.
    
	

\section{Fidelity distribution for Schur and unitary channels}\label{sec:proposed_method}
    Consider Schur channels of an arbitrary dimension $d$. Our goal here is to fully characterize the distribution $P\left(F\right)$ of their operation fidelity.
    Such channels can be described by Kraus operators which are Hermitian and mutually commuting.  
    Provided at least $(d-1)$ Kraus operators
    are orthogonal, we are able to find the full distribution of the operation fidelity.     
    Later on, in Subsection~\ref{subsec:arbitrary_unitary_channel}, these results are used to develop a method for determining the distribution of the operation fidelity $P\left(F\right)$ for a unitary channel of an arbitrary size.
    
    \subsection{Schur channels}\label{subsec:Schur_channels}    
    Let us show that the joint numerical range of a sufficiently high number of orthogonal, commuting Hermitian operators forms a simplex.
    
    \begin{lemma}\label{lemma:linearly_independent_simplex}
    For a set of $(d-1)$ commuting Hermitian operators $\vec{H} = \{H_i\}_{i=1}^{d-1}$ of size $d$, the following statements are equivalent:
    \begin{enumerate}
        \item[(1)] $\{H_i\}\cup\mathbb{I}_d$ are linearly independent, i.e.\ their diagonalized forms span the entire space of diagonal matrices,
        \item[(2)] their joint numerical range $W(\vec{H})$ forms a non-degenerate simplex of dimension $\left(d-1\right)$.
    \end{enumerate}
    \end{lemma}

    The proof is provided in Appendix~\ref{app:proofs}.
    Note that in the above lemma, the condition that $\mathrm{span}\{H_i\}$ does not contain $\mathbb{I}_d$ is necessary, since otherwise one of the dimensions of the simplex would be degenerate. 
    As an example, consider $d=2$ with a single Hermitian matrix $H_1 = \mathrm{diag}(1,1)$. 
    Then, the numerical range is a point -- a degenerated 1-simplex (interval). 
    
    Consequently, we are able to prove statements concerning the uniformity of the joint numerical shadow for a set of commuting Hermitian operators.
    
    \begin{theorem}\label{thm:hermitian_operators}
         Let $\vec{H} = \{H_i\}_{i=1}^{d-1}$ be a set of $(d-1)$ commuting Hermitian operators of size $d$ such that the joint numerical range forms a $(d-1)$-dimensional simplex $\Delta$ (or, equivalently, such that $\{H_i\}\cup\mathbb{I}_d$ are linearly independent).
        Then, the joint numerical shadow $P_{\vec{H}}(\vec{r})$ is uniform in $\Delta$.
    \end{theorem}
    
    This result can be extended to any number $m$ of commuting Hermitian operators, provided $m\geq d-1$ and at least $(d-1)$ of them are linearly independent. 
    
    \begin{theorem}\label{thm:generalized_hermitian_operators}
        Let $ {\vec H}=\{H_i\}_{i=1}^m$ be a set of commuting Hermitian matrices of size $d$, diagonalizable by $U$, such that $m\geq d-1$ and $\{UH_iU^\dagger\}\cup \mathbb{I}_d$ spans the whole vector space of diagonal matrices of size $d$. Then, the joint numerical shadow $P_{\vec{H}}(\vec{r})$ of $\vec H$, supported in the simplex $\Delta\subset\mathds{R}^m$, is uniform.
    \end{theorem}
    
    The proofs of both theorems, based on Lemma~\ref{lemma:linearly_independent_simplex}, are given in Appendix~\ref{app:proofs}.
    These statements allow us to design a method for finding the analytical formula for the distribution of the operation fidelity for $m$ orthogonal, commuting $d$-dimensional Hermitian Kraus operators that satisfy the assumptions of Theorem~\ref{thm:generalized_hermitian_operators}.
    
    First, we construct a $d$-point regular simplex $\Delta$ formed by their joint numerical range. 
    Then, we perform an integration~(\ref{eq:fid_dist_from_nom_shad}) over the simplex.
    However, due to the uniformity of the shadow $P_{\vec{H}}\left(\vec{r}\right)=\mathrm{vol}(\Delta)^{-1}$ (the inverse of the volume of the simplex), the integration is significantly simplified. 
    Thus, we arrive at the 
    last statement of this subsection.

    \begin{corollary}
        Suppose that $\Phi$ is a channel with $m$ orthogonal, commuting Hermitian Kraus operators of size $d$ that satisfy the assumptions of Theorem~\ref{thm:generalized_hermitian_operators}.
        Then, the distribution of the operation fidelity of $\Phi$ is given by an integration over the joint numerical range, which forms the $d$-point simplex $\Delta \subset {\mathbb R}^m$,
        \begin{equation}\label{eq:fid_dist_uniform_simplex}
	    P(F) = \frac{1}{\mathrm{vol}(\Delta)} \int_{\Delta} \delta(F-|\vec{r}\left(\vec{x}\right)|^2)\mathrm{d}^{d-1}x,
	\end{equation}
    where $\vec{x}$ parametrizes the $\left(d-1\right)$-dimensional affine subspace containing $\Delta$ and 
    $\vec{r}\left(\vec{x}\right)$ 
    represents the corresponding point in $\mathds{R}^m$.
    \end{corollary}

    A convenient representation of probability distributions corresponding to numerical shadows of some operators
    can be obtained with help of the notion of basis spline, also called {\sl B-spline}~\cite{Dunkl_2015} --- see Appendix~\ref{app:b-spline}.
    Such splines are linked to the generalized Dirichlet distributions, from which real shadows of real, symmetric matrices can be obtained.
    In an analogous way one can express the 
    distribution~(\ref{eq:fid_dist_uniform_simplex}) 
    by the derivative of the relative volume of 
    the intersection of a ball and a simplex, 
 \begin{equation}
    P\left(F\right)=\frac{\mathrm{d}}{\mathrm{d}F}
        \frac{\mathrm{vol}\left(B_F \cap \Delta \right)}{\mathrm{vol}\left(\Delta\right)},
    \end{equation}
    where $B_F=B\left(\vec{0},\sqrt{F}\right)$ represents the $m$-dimensional ball of radius $\sqrt{F}$, with the center at the point $\vec{0}\in\mathds{R}^m$ and immersed in the same $m$-dimensional space as the $d$-point simplex $\Delta$. 

    Finally, in Fig.~\ref{fig:examples_4-level_PDF} we use Eq.~(\ref{eq:fid_dist_uniform_simplex}) to visualize the operation fidelity distributions for exemplary quantum channels that satisfy the assumptions of Theorem~\ref{thm:generalized_hermitian_operators}.

    \onecolumngrid
    
    \begin{figure}[H]
        \begin{tikzpicture}[scale=0.885, every node/.style={scale=0.885}]
    \node (a) at (0,0) {\includegraphics[width=0.48\textwidth]{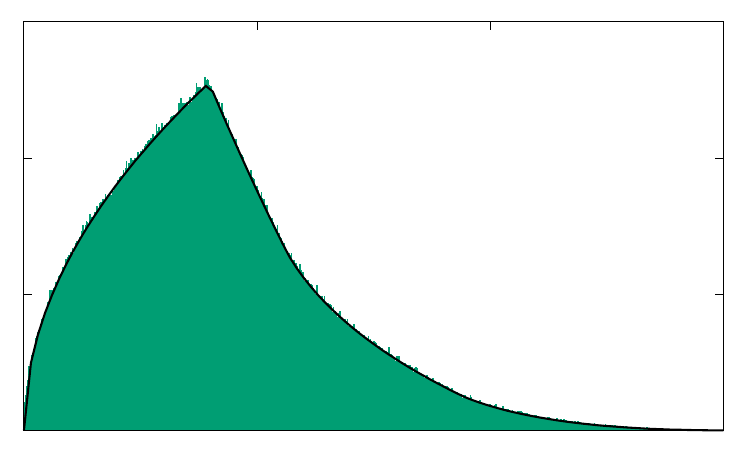}};
    \node (a) at (0,+2.8) {(a)};
    \node (a) at (0,-2.8) {$F$};
    \node (a) at (-4,-2.7) {$0$};
    \node (a) at (-1.33,-2.7) {$0{.}33$};
    \node (a) at (1.33,-2.7) {$0{.}67$};
    \node (a) at (4,-2.7) {$1$};
    \node (a) at (-4.5,0) {\begin{turn}{90}$P(F)$\end{turn}};
    \node (a) at (-4.3,-2.3) {$0$};
    \node (a) at (-4.5,-0.77) {$1{.}17$};
    \node (a) at (-4.5,0.77) {$2{.}33$};
    \node (a) at (-4.4,2.3) {$3{.}5$};
\end{tikzpicture}
        \hspace*{\fill}
        \begin{tikzpicture}[scale=0.885, every node/.style={scale=0.885}]
    \node (a) at (0,0) {\includegraphics[width=0.48\textwidth]{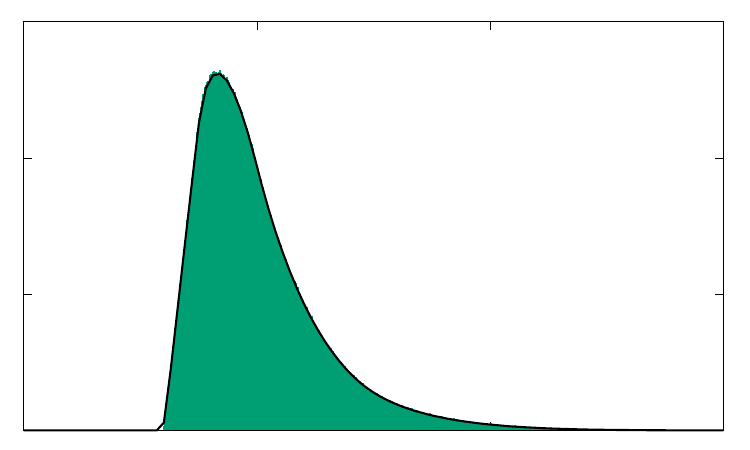}};
    \node (a) at (0,+2.8) {(b)};
    \node (a) at (0,-2.8) {$F$};
    \node (a) at (-4,-2.7) {$0$};
    \node (a) at (-1.33,-2.7) {$0{.}33$};
    \node (a) at (1.33,-2.7) {$0{.}67$};
    \node (a) at (4,-2.7) {$1$};
    \node (a) at (-4.5,0) {\begin{turn}{90}$P(F)$\end{turn}};
    \node (a) at (-4.3,-2.3) {$0$};
    \node (a) at (-4.5,-0.77) {$2{.}33$};
    \node (a) at (-4.5,0.77) {$4{.}67$};
    \node (a) at (-4.3,2.3) {$7$};
\end{tikzpicture}
        
        \caption{Distribution $P\left(F\right)$ of the operation fidelity for exemplary Schur channels: (a) acting on a four-level system and defined by three diagonal Kraus operators $K_1=\frac{1}{\sqrt{2}}\mathrm{diag}\left(1,\frac{1}{2},-\frac{1}{2},-1\right)$, $K_2=\frac{\sqrt{3}}{2\sqrt{2}}\mathrm{diag}\left(0,1,1,0\right)$ and $K_3=\frac{1}{\sqrt{2}}\mathrm{diag}\left(1,-1,1,-1\right)$, and (b) acting on a five-level system and defined by orthogonal projection operators $K_j=\ket{j}\bra{j}$ for $j\in\left\{0,1,2,3,4\right\}$. Solid curves denote distributions obtained using Eq.~(\ref{eq:fid_dist_from_nom_shad}) with uniform numerical shadow and underlying histograms refer to numerical data.}
        \label{fig:examples_4-level_PDF}
    \end{figure}
    \twocolumngrid
    
    \subsection{Unitary channels}\label{subsec:arbitrary_unitary_channel}
    Theorems proved in the previous subsection enable us to introduce a method for determining the distribution of the operation fidelity for a unitary channel $U$ of an arbitrary dimension $d$. 
    Since any unitary operator is normal, its numerical range $W\left(U\right)$ is equal to the convex hull of its spectrum, which forms a polygon inscribed in the unit circle.
    As before, we shall base our method on the fact that the corresponding numerical shadow, $P_U\left(z\right)$, is a probability measure obtained by projecting the uniform measure on a $\left(d-1\right)$-simplex into the polygon $W\left(U\right)$. 
    
    We shall denote the Hermitian and the anti-Hermitian part of $U=U_H + iU_A$ as $U_H$ and $U_A$, respectively.
    Suppose that $\{H_i\}$ is a collection of $m$ Hermitian operators of size $d$, commuting with $U$, such that the set $\{U_H,U_A,\mathbb{I}_d\}\cup\{H_i\}$ satisfies the conditions of Theorem~\ref{thm:generalized_hermitian_operators} -- i.e.\ after a simultaneous diagonalization, they span the space of diagonal matrices.
    
    Then, applying Theorem~\ref{thm:generalized_hermitian_operators}, we observe that the joint numerical shadow $P_{U_H,U_A,H_1,...}(\vec{r})$ is uniform on the $d$-point simplex. 
    Therefore, to obtain the operation fidelity of the channel defined as $\Phi(\rho) = U\rho U^\dagger + \sum_i H_i \rho H^\dagger_i$ it suffices to use Proposition~\ref{prop:kraus_fidelity}. 
    
    Now, the crucial point is that this setup holds independently of the norms of matrices $\{H_i\}$.
    Thus, to find the numerical shadow of $U$ together with the distribution of operation fidelity we introduce a one-parameter family of channels $\Phi_\varepsilon$, defined as
    \begin{equation}\label{eq:aux_channel}
        \Phi_\varepsilon\left(\rho\right)=\left(1-\varepsilon\right)U\rho U^\dagger+\varepsilon\sum_{i=1}^m H_i \rho H_i^\dagger, 
    \end{equation}
    where $0\leq\varepsilon\leq 1$. 
    Then, by taking a limit $\varepsilon\rightarrow 0$, we obtain the channel that resembles the original unitary one, $\lim_{\varepsilon\to 0} \Phi_\varepsilon (\rho ) = U\rho U^\dagger$. 

    Simultaneously, for any $\varepsilon > 0$, the joint numerical shadow will be uniform. 
    Therefore, to find the distribution of the operation fidelity, we use Theorem~\ref{thm:generalized_hermitian_operators} for $\Phi_\varepsilon$, which gives the desired distribution of the original unitary channel with an arbitrary precision, limited by $\varepsilon$. 
    The convergence to the desired distribution of the operation fidelity for an exemplary unitary channel $U$ is shown in Fig.~\ref{fig:conv_ex}. 
    
    As a side note, observe that the above setup can be applied to any normal matrix to obtain its numerical shadow by taking sufficiently many commuting Hermitian operators. 
    Similarly, the joint numerical shadow for any set of commuting normal operators can be derived as well, i.e.\ for a set of $m$ Hermitian matrices, where $m\geq d-1$.

    \begin{center}
    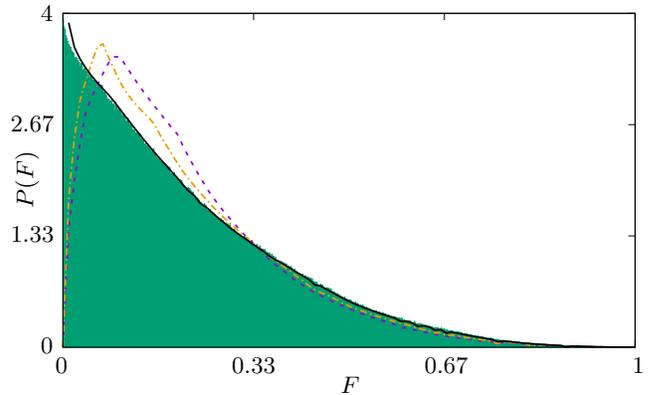
\begin{figure}[H]
	\centering
	\begin{tikzpicture}[scale=0.65, every node/.style={scale=0.65}]
	\node (a) at (0,0) {\includegraphics[width=0.7\textwidth]{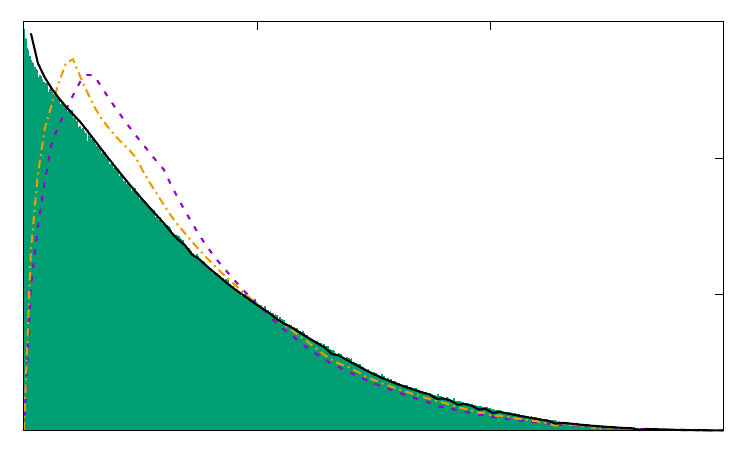}};
	\node (a) at (0,-4.2) {\scalebox{1.54}{$F$}};
    \node (a) at (-5.9,-3.8) {\scalebox{1.54}{$0$}};
	\node (a) at (-1.97,-3.8) {\scalebox{1.54}{$0{.}33$}};
	\node (a) at (1.97,-3.8) {\scalebox{1.54}{$0{.}67$}};
	\node (a) at (5.9,-3.8) {\scalebox{1.54}{$1$}};
	\node (a) at (-6.6,0) {\scalebox{1.54}{\begin{turn}{90}$P(F)$\end{turn}}};
	\node (a) at (-6.2,-3.4) {\scalebox{1.54}{$0$}};
	\node (a) at (-6.5,-1.13) {\scalebox{1.54}{$1{.}33$}};
	\node (a) at (-6.5,1.13) {\scalebox{1.54}{$2{.}67$}};
    \node (a) at (-6.2,3.4) {\scalebox{1.54}{$4$}};
\end{tikzpicture}
    \caption{Distribution $P\left(F\right)$ of operation fidelity for a $4$-dimensional channel defined by a unitary operator $U = \text{diag}(1,\mathrm{e}^{\frac{\mathrm{i}\pi}{3}},-1,-\mathrm{i})$, shown using green bars, obtained numerically. Operation fidelity distributions evaluated for auxiliary channels $\widetilde{\Phi}$ defined by Eq.~(\ref{eq:aux_channel}), with $m=1$ and $H_1=\text{diag}(1,-1,1,-1)$, are shown for different values of $\varepsilon$ ($0{.}2$ for the purple, dashed line, $0{.}1$ for the yellow, dot-dashed line and $0{.}01$ for the black, solid line) and obtained using Eq.~(\ref{eq:fid_dist_uniform_simplex}). 
    Distribution $P(F)$ for $\widetilde{\Phi}$ converges to the operation fidelity distribution of $U$ as $\varepsilon$ approaches $0$.}\label{fig:conv_ex}
    \end{figure}
    \end{center}
	
	\section{Conclusions} 
	Generalized quantum channels are not easy to characterize. 
	Here, we expanded upon one of their aspects; namely, the operational fidelity -- one of the measures of distance to the identity channel. 
	Interestingly, recalling an important notion from linear algebra -- numerical range -- we were able to find the distributions $P\left(F\right)$ of operation fidelity in various low-dimensional cases. We discussed the cases of single-qubit mixed unitary channels and single-qutrit unitary channels.
	
	Furthermore, we characterized the Schur channels for an arbitrary dimension $d$, as well as provided a method for determining the distribution of operation fidelity of a unitary channel in the general case.
	Additionally, extremal fidelities $F_{\text{min}}$ and $F_{\text{max}}$ in the general cases of Schur channels and single-qubit mixed unitary channels were also described.
	
	These results widen our understanding of quantum channels and can be used for channel estimation and discrimination. 
    We propose a novel setup for this task that utilizes fidelity between the unknown initial and final state, as discussed in Subsection~\ref{subsection:probabilistic_estimation}.
	
	Naturally, several questions ensue our research.
	First of all, in which classes of channels in an arbitrary dimension can the operation fidelity distribution be found?
	It is possible that the generalizations of Pauli channels discussed in Section~\ref{sec:PDF_for_qubit_qutrit}, such as the Weyl-Heisenberg channels, shall also admit simple characterization. 
	
	An estimate of how many samples are needed to have a sufficiently high confidence level for channel estimation would be also interesting, with a wide variety of different possible setups, e.g.\ two channels, family indexed by a parameter, etc.
	Finally, it is important to apply the results of the present paper to real-world devices. 
	In order to do so, channels that occur in a laboratory should be parametrized in a manner facilitating their description using our setup.
	We hope that both theoretical and experimental progress will follow due to their conceivable importance for present-era quantum devices.

    We are grateful to Chi-Kwong Li, \L{}ukasz Pawela and Zbigniew Pucha\l{}a for fruitful discussions.
    Financial support by Narodowe Centrum Nauki under the Maestro grant number DEC-2015/18/A/ST2/00274 and by Foundation for Polish Science under the Team-Net project no.\ POIR.04.04.00-00-17C1/18-00 is gratefully acknowledged. 
    GRM additionally acknowledges support from 
    ERC AdG NOQIA; Ministerio de Ciencia y Innovation Agencia Estatal de Investigaciones (PGC2018-097027-B-I00/10.13039/501100011033, CEX2019-000910-S/10.13039/501100011033, Plan National FIDEUA PID2019-106901GB-I00, FPI, QUANTERA MAQS PCI2019-111828-2, QUANTERA DYNAMITE PCI2022-132919, Proyectos de I+D+I “Retos Colaboración” QUSPIN RTC2019-007196-7); MCIN Recovery, Transformation and Resilience Plan with funding from European Union NextGenerationEU (PRTR C17.I1); Fundació Cellex; Fundació Mir-Puig; Generalitat de Catalunya (European Social Fund FEDER and CERCA program, AGAUR Grant No. 2017 SGR 134, QuantumCAT \ U16-011424, co-funded by ERDF Operational Program of Catalonia 2014-2020); Barcelona Supercomputing Center MareNostrum (FI-2022-1-0042); EU Horizon 2020 FET-OPEN OPTOlogic (Grant No 899794); EU Horizon Europe Program (Grant Agreement 101080086 — NeQST), National Science Centre, Poland (Symfonia Grant No. 2016/20/W/ST4/00314); European Union’s Horizon 2020 research and innovation program under the Marie-Skłodowska-Curie grant agreement No 101029393 (STREDCH) and No 847648 (“La Caixa” Junior Leaders fellowships ID100010434: LCF/BQ/PI19/11690013, LCF/BQ/PI20/11760031, LCF/BQ/PR20/11770012, LCF/BQ/PR21/11840013). Views and opinions expressed in this work are, however, those of the author(s) only and do not necessarily reflect those of the European Union, European Climate, Infrastructure and Environment Executive Agency (CINEA), nor any other granting authority. Neither the European Union nor any granting authority can be held responsible for them.
    
    \appendix
    \section{Different channels with the same operation fidelity distribution}\label{app:same_fidelity}	
    Here, we prove the statement from Subsection~\ref{subsection:probabilistic_estimation} concerning channels with the same distribution of the operation fidelity. 
	\begin{lemma}
	Let $\Phi$ and $\Phi'$ be 2 quantum channels acting on a $d$-dimensional Hilbert space $\mathbb{C}^d$, each defined by $m$ Kraus operators $\left\{K_j\right\}_{j=1}^m=\left\{H_j+\mathrm{i}A_j\right\}_{j=1}^m$ and $\left\{K'_j\right\}_{j=1}^m=\left\{H'_j+\mathrm{i}A'_j\right\}_{j=1}^m$, where $H_j$, $A_j$, $H'_j$ and $A'_j$ are Hermitian for every $j$. We define $\widetilde{K}=\left(H_1,H_2,\ldots,H_m,A_1,A_2,\ldots,A_m\right)$ and $\widetilde{K}'=\left(H'_1,H'_2,\ldots,H'_m,A'_1,A'_2,\ldots,A'_m\right)$. Then if
	\begin{equation}
	    \widetilde{K}'_i=\sum_{j=1}^{2m}O_{ij}U^\dagger\widetilde{K}_j U
	\end{equation}
	for a certain orthogonal $O\in\mathrm{O}\!\left(2m\right)$ and unitary $U\in\mathrm{U}\!\left(d\right)$, the channels $\Phi$ and $\Phi'$ share the same distribution $P\left(F\right)$ of operation fidelity.
	\end{lemma}
	\begin{proof}
    Using Eq.~(\ref{eq:def_joint_numerical_shadow}) and Eq.~(\ref{eq:fid_dist_from_nom_shad}), we express the operation fidelity distribution of channel $\Phi'$ as
    \begin{equation}
        P'\left(F\right)=\int_{\Omega_d}\delta\left(F-\sum_{i=1}^{2m}\bra{\psi}\widetilde{K}'_i\ket{\psi}^2\right)\mathrm{d}\mu\left(\psi\right).
    \end{equation}
    Since $\widetilde{K}'_i=\sum_{j=1}^{2m}O_{ij}U\widetilde{K}_j U^\dagger$, and due to the orthogonality of $O$
    \begin{equation}
        P'\left(F\right)=\int_{\Omega_d}\delta\left(F-\sum_{i=1}^{2m}\bra{\psi}U^\dagger\widetilde{K}_i U\ket{\psi}^2\right)\mathrm{d}\mu\left(\psi\right).
    \end{equation}
    Furthermore, $\mathrm{d}\mu\left(\psi\right)$ is a measure on the set $\Omega_d$ of pure states that is invariant under unitary transformations. Thus, we can change variables from $\ket{\psi}$ to $\ket{\phi}=U\ket{\psi}$.
    This implies that both $\Phi$ and $\Phi'$ have the same operation fidelity distributions.
	\end{proof}

\section{Proofs of theorems from Section~\ref{sec:proposed_method}}\label{app:proofs}
For the convenience of the reader we rewrite here the three statements made in the main body of the text and then prove them.
First result from Section~\ref{sec:proposed_method} concerns the connection between the orthogonality of commuting Hermitian operators and geometrical properties of their joint numerical shadow.

\begin{manuallemma}{\ref{lemma:linearly_independent_simplex}}
    For a set of $(d-1)$ commuting Hermitian operators $\{H_i\}_{i=1}^{d-1}$ of size $d$, the following statements are equivalent:
    \begin{enumerate}
        \item[(1)] $\{H_i\}\cup\mathbb{I}_d$ are linearly independent, i.e.\ their diagonalized forms span the entire space of diagonal matrices,
        \item[(2)] their joint numerical range $W(H_i)$ forms a non-degenerate simplex of dimension $(d-1)$.
    \end{enumerate}
\end{manuallemma}
\begin{proof}
    It suffices to verify their linear independence in basis in which all $\{H_i\}$ are diagonal.
    To do that, we shall use the previous notation, where $\{{\lambda_{ij}}\}_{j=1}^{d}$ is the set of eigenvalues of matrix $H_i$. 
    Let us start with the direction $(2) \implies (1)$. 
    
    The joint numerical range is spanned by $d$ points $\{\lambda_{1j},...,\lambda_{(d-1)j}\}_{j=1}^d$ that form a $(d-1)$-simplex $\Delta$. 
    Assuming that the simplex is non-degenerate, it has a non-zero volume $\mathrm{vol}\, (\Delta)$ given by the determinant of a $d\times d$ matrix~\cite{Stein_1966}
    \begin{equation}\label{eq:independent_volume_simplex}
        \mathrm{vol}\, (\Delta) = \frac{1}{(d-1)!} \,\left|\mathrm{det} 
        \begin{pmatrix}
                \lambda_{11} & ... & \lambda_{1d}\\
                \vdots & \ddots & \vdots \\
                \lambda_{(d-1)1} & ... & \lambda_{(d-1)d} \\ 
                1 & ... & 1
        \end{pmatrix}\right|.
    \end{equation}
    Since the volume is non-zero, all rows are linearly independent. 
    However, these rows form diagonals of $\{H_i\}\cup\mathbb{I}_d$ matrices, making them also linearly independent. 
    
    Direction $(1) \implies (2)$ is straightforward, since linearly independent matrices will form first $(d-1)$ rows of the matrix from Eq.~(\ref{eq:independent_volume_simplex}), while the last row is formed by the identity matrix $\mathbb{I}_d$.
    Therefore, the determinant is non-zero, making the simplex non-degenerate.
\end{proof}

Then, we use the above lemma to prove the uniformity of the joint numerical shadow for the case of orthogonal, commuting Hermitian operators.

\begin{manualtheorem}{\ref{thm:hermitian_operators}}
    Let $\vec{H} = \{H_i\}_{i=1}^{d-1}$ be a set of $(d-1)$ commuting Hermitian operators of size $d$ such that the joint numerical range forms a $(d-1)$-dimensional simplex $\Delta$ (or, equivalently, such that $\{H_i\}\cup\mathbb{I}_d$ are linearly independent).
    Then, the joint numerical shadow $P_{\vec{H}}(\vec{r})$ is uniform in $\Delta$.
\end{manualtheorem}
\begin{proof}
    Let us start by recalling Eq.~(\ref{eq:def_joint_numerical_shadow}), which defines the joint numerical shadow as an integral over the set $\Omega_d$ of all pure states of dimension $d$,
    \begin{equation}
        P_{\vec{H}}(\vec{r})=\int_{\Omega_d}\delta\left(\vec{r}-\bra{\psi}\vec{H}\ket{\psi}\right)\mathrm{d}\mu\left(\psi\right),
    \end{equation}
    where $\bra{\psi}\vec{H}\ket{\psi}$ denotes an $(d-1)$-dimensional vector $\left(\bra{\psi}H_1\ket{\psi},...,\bra{\psi}H_{d-1
    }\ket{\psi}\right)$. 
    Operators $H_i$ commute, hence without loss of generality we may use the basis $\ket{j}$ in which these operators are diagonal. 
    Following Section~\ref{sec:diagonal_Kraus}, we denote their eigenvalues by $\lambda_{ij}$ so that
    $H_i = \sum_{j=1}^d \lambda_{ij} \ket{j}\bra{j}$. 
    
    Decomposing the state $\ket{\psi} = \sum_j c_j \ket{j}$, we arrive at the following expression
    \begin{equation}
    \begin{split}
        &P_{\vec{H}}(\vec{r}) = \\ &=\int_{\Omega_d}\delta \left(\vec{r} - \left[\sum_{j=1}^d |c_j|^2 \lambda_{1j},...,\sum_{j=1}^d |c_j|^2 \lambda_{(d-1)j}\right] \right) \mathrm{d}\mu\left(\psi\right).
    \end{split}
    \end{equation}
    However, since squared absolute values of coefficients in any basis are distributed as a probability vector uniform in the probability simplex $\Delta_{d-1} = \{[x_1,...,x_d]|\sum_{j=1}^d x_j = 1 \land x_j\geq 0\}$~\cite{Dunkl_2011}, we can change the integration domain from pure states with the Haar measure into a simplex parametrized with $(d-1)$ variables $\vec{x} = \{x_j\}_{j=1}^{d-1}$, with a uniform measure on the simplex
    \begin{widetext}
    \begin{equation}
        P_{\vec{H}}(\vec{r}) = \int_{\Delta_{d-1}}\delta \left(\vec{r} - \left[\sum_{j=1}^{d-1} x_j \lambda_{1j}+ \bigg(1-\sum_{j=1}^{d-1} x_j\bigg)\lambda_{1d},...,\sum_{j=1}^{d-1} x_j \lambda_{(d-1)j} + \bigg(1-\sum_{j=1}^{d-1} x_j\bigg)\lambda_{(d-1)d}\right] \right) \mathrm{d}\vec{x}.
    \end{equation}        
    \end{widetext}
    
    Subsequently, we define $(d-1)$ new variables $y_i = \sum_{j=1}^{d-1} x_j \lambda_{ij}+ \big(1-\sum_{j=1}^{d-1} x_j\big)\lambda_{id}$. 
    The Jacobian matrix between variables $\vec{x}$ and $\vec{y}$ reads $J_{ij} = \frac{\partial y_i}{\partial x_j} = \lambda_{ij}-\lambda_{id}$, thus it is independent of $\vec{y}$ 
    \begin{equation}\label{eq:final_shadow_theorem}
    \begin{split}
        P_{\vec{H}}(\vec{r}) &= \int_{\Delta}\delta \left(\vec{r} - \left[y_1,...,y_{d-1}\right] \right)|\textrm{det}\, J^{-1}| \mathrm{d}\vec{y} = \\ &=|\textrm{det}\, J^{-1}|\int_{\Delta}\delta \left(\vec{r} - \left[y_1,...,y_{d-1}\right] \right) \mathrm{d}\vec{y}.     
    \end{split}      
    \end{equation}
    Using elementary operations on matrices, it is easy to prove that the Jacobian is proportional to the volume of the simplex, given by Eq.~(\ref{eq:independent_volume_simplex}), formed by the eigenvalues of the Hermitian matrices under consideration~\cite{Stein_1966}, $|\mathrm{det}\, J| = (d-1)! \, \mathrm{vol}_{d-1}\, \big(\Delta\big)$.
    Consequently, since we assumed that $\{H_i\}\cup\mathbb{I}_d$ are linearly independent and $\Delta$ is a non-degenerate $(d-1)$-simplex (Lemma~\ref{lemma:linearly_independent_simplex}), the invertibility of Jacobian matrix is assured.
    
    Finally, we observe that the value of the integral on the right-hand side of Eq.~(\ref{eq:final_shadow_theorem}) equals 1 if and only if $\vec{r}$ belongs to $\Delta$, i.e.\ the joint numerical range of operators $\{H_i\}$. 
    Therefore, due to the independence of the Jacobian of $\vec{r}$, we conclude that the joint numerical shadow is uniform in the simplex $\Delta$.     
\end{proof}

Finally, Theorem~\ref{thm:hermitian_operators} can be generalized by adding more commuting Hermitian operators to a set of operators that already satisfies the assumptions of the said theorem.

\begin{manualtheorem}{\ref{thm:generalized_hermitian_operators}}
    Let $ {\vec H}=\{H_i\}_{i=1}^m$ be a set of commuting Hermitian matrices of size $d$, diagonalizable by $U$, such that $m\geq d-1$ and $\{UH_iU^\dagger\}\cup \mathbb{I}_d$ spans the whole vector space of diagonal matrices of size $d$. Then, the joint numerical shadow $P_{\vec{H}}(\vec{r})$ of $\vec H$, supported in the simplex $\Delta\subset\mathds{R}^m$, is uniform.
\end{manualtheorem}
\begin{proof}
    First of all, let us observe that $m\geq d-1$. 
    Therefore, we can split the set of Hermitian matrices into a part $\{H_i\}_{i=1}^{d-1}$ which, together with the identity matrix, will span the space of diagonal matrices and the rest $\{H_i\}_{i=d}^{m}$. 
    Then, the Dirac delta in the formula for the joint numerical shadow can be separated into two parts:
    \begin{widetext}
    \begin{equation}
    \begin{split}
        P_{\vec{H}}(\vec{r}) &= \int_{\Omega_d}\delta \left(\vec{r} - \left[\sum_{j=1}^d |c_j|^2 \lambda_{1j},...,\sum_{j=1}^d |c_j|^2 \lambda_{mj}\right] \right) \mathrm{d}\mu\left(\psi\right) = \\
        &= \int_{\Omega_d}\delta\left(\vec{r_1} - \left[\sum_{j=1}^d |c_j|^2 \lambda_{1j},...,\sum_{j=1}^d |c_j|^2 \lambda_{(d-1)j}\right] \right)
        \delta\left(\vec{r_2} - \left[\sum_{j=1}^d |c_j|^2 \lambda_{dj},...,\sum_{j=1}^d |c_j|^2 \lambda_{mj}\right] \right)\mathrm{d}\mu\left(\psi\right),
    \end{split}
    \end{equation}
    \end{widetext}
    where $m$-dimensional vector $\vec{r}$ was split into two parts, $\vec{r} = [\vec{r_1},\vec{r_2}]$.
    Now, we can perform this integration by utilizing techniques used in the proof of Theorem~\ref{thm:hermitian_operators}. In this way, we obtain a formula of the following form:
    \begin{equation}
        P_{\vec{H}}\left(\vec{r}\right)=P_{\{H_i\}_{i=1}^{d-1}}\left(\vec{r}_1\right)\delta\left(\vec{r}_2-\vec{R}_2-T\vec{r}_1\right),
    \end{equation}
    where $\vec{R}_2\in\mathds{R}^{m-d+1}$ is a certain constant vector, $T$ -- a certain linear transformation from $\mathds{R}^{d-1}$ to $\mathds{R}^{m-d+1}$ and $P_{\{H_i\}_{i=1}^{d-1}}\left(\vec{r}_1\right)$ -- numerical shadow of the collection of operators $\{H_i\}_{i=1}^{d-1}$. Therefore the, numerical shadow of a collection of operators $\{H_i\}_{i=1}^m$ is uniform in the numerical range of that collection, which is contained in the affine subspace of $\mathds{R}^m$ defined by $\vec{r}_2=\vec{R}_2+T\vec{r}_1$.
\end{proof}

\section{Definition of B-spline}\label{app:b-spline}
For the convenience of the reader we recall here the definition of the measure version of B-spline, as given in~\cite{Dunkl_2015}.

Let $A\subset\mathds{R}^m$ be a set and $\Delta\subset\mathds{R}^m$ -- a non-degenerate simplex. The measure version of B-spline of order $0$ is then defined by the following expression:
    \begin{equation}
          \mathcal{M}_{\Delta}
         \left(A\right)=\frac{\mathrm{vol}\left(\Delta\cap A\right)}{\mathrm{vol}\left(\Delta\right)},
    \end{equation}
therefore it corresponds exactly to the relative volume of the intersection of the set $A$ and the $d$-point simplex $\Delta$ with respect to that simplex.
        
	\bibliographystyle{unsrt}
	\bibliography{bibliography}
\end{document}